\documentclass[a4paper,10pt]{article}
\usepackage{amsfonts}
\usepackage{enumitem}
\usepackage{amsthm}
\usepackage{nccmath}
\usepackage{amsmath}
\usepackage{amssymb}
\usepackage{setspace}

\advance\textwidth by 50pt
\advance\oddsidemargin by -25pt
\advance\evensidemargin by -25pt

\newtheorem{lemma}{Lemma}[section]
\newtheorem{proposition}[lemma]{Proposition}
\newtheorem{theorem}[lemma]{Theorem}
\newtheorem{definition}{Definition}[section]
\theoremstyle{definition}
\newtheorem{remark}[lemma]{Remark}

\newcommand{\Intxt}[1]{\int_{[0, \tau ) \times \Omega } \! #1 \, dtd\mathbf x}        
\newcommand{\IntXt}[1]{\int_{[0, \tau ) \times \Omega } \! #1 \, dtd\mathbf X}        
\newcommand{\Intx}[1]{\int_{\Omega } \! #1 \, d \mathbf x}                            
\newcommand{\wt}{\widetilde{\mathbf w}}                      

\begin{document}

\vspace{3mm}
\begin{center}
\Large
{\bf Solutions of the Fully Compressible Semi-Geostrophic System}

\vspace{3mm}

\large
M.J.P. Cullen$^{*}$, D.K. Gilbert  and B. Pelloni

\vspace{5mm}
{\em 
Department of Mathematics and Statistics,

University of Reading, 

Reading RG6 6AX, UK}

\vspace{2mm}
$^{*}$ also {\em Met Office, Fitzroy Road, 

Exeter EX1 3PB, UK}

\vspace{3mm}
\today
\end{center}
\normalsize

\begin{abstract}
The fully compressible semi-geostrophic system is widely used in the modelling of large-scale atmospheric flows. In this paper, we prove rigorously the existence of weak Lagrangian solutions of  this system,  formulated in the original physical coordinates.  In addition, we provide an alternative  proof of the earlier result on the existence of weak solutions of this system expressed in the so-called geostrophic, or dual, coordinates. The proofs are based on the optimal transport formulation of the problem and on recent general results concerning transport problems posed in the Wasserstein space of probability measures.
\end{abstract}

{\em Keywords:} geophysical fluid mechanics, nonlinear PDEs, meteorology. 

\smallskip
{\em MSC classification:}  35Q35, 35Q86

\section{Introduction}\label{introduction}

The behaviour of the atmosphere is governed by the compressible Navier-Stokes equations, together with the laws of thermodynamics, equations describing phase changes, source terms and boundary fluxes.  
These equations are too complex to be solved accurately in
a large-scale atmospheric context and therefore reductions and approximations of the
Navier-Stokes equations are often used to validate and understand the solutions that have been computed.

One such approximation, valid on scales where the effects of rotation dominate the flow, is the system of semi-geostrophic equations, see equations (\ref{commom})-(\ref{comstate}) below.  First introduced by Eliassen \cite{eliassen} and then rediscovered by Hoskins \cite{hoskins}, the semi-geostrophic equations are an extremely useful model, particularly in describing the formation of fronts, and are both  studied theoretically and used widely for numerical modelling and simulations. 

In this paper, we prove rigorously the existence of weak Lagrangian solutions of the fully compressible semi-geostrophic system. The proof combines ideas of several previous papers, starting with the pioneering work of Benamou and Brenier \cite{benamou} on the formulation of the semi-geostrophic equations as an optimal transport problem, and modern methods in the analysis of Hamiltonian ODEs in probability spaces, in particular the work of Ambrosio \cite{ambgangbo}, \cite{ambbook}.

For an accurate representation of  the behaviour of large-scale atmospheric flow, one should consider the fully compressible semi-geostrophic equations with variable Coriolis parameter and a free upper boundary condition. The complexity of this problem means that so far results have only been obtained after relaxing one or more of these conditions.

In \cite{benamou}, Benamou and Brenier assumed the fluid to be  incompressible, the Coriolis parameter  constant and the boundaries rigid. They then used a change of variables, introduced by Hoskins in \cite{hoskins}, to derive the so-called dual formulation. In this formulation,  the equations are interpreted as a Monge-Amp\`{e}re equation coupled with a transport problem,  to prove the existence of  stable weak solutions. 

In \cite{gangbo}, Cullen and Gangbo relaxed the assumption of rigid boundaries with a more physically appropriate free boundary condition.  However, they additionally assumed constant potential temperature to obtain the 2-D system known as the semi-geostrophic shallow water system. After passing to dual variables, they showed existence of stable weak solutions for this system of equations.

In \cite{maroofi}, Cullen and Maroofi proved that stable weak solutions of the 3-D compressible semi-geostrophic system  in its dual formulation exist, returning to the assumption of rigid boundaries.  

The main problem posed by the existence results in \cite{benamou}, \cite{gangbo} and \cite{maroofi} is that they are all proved in dual space.  It is difficult to relate these results directly to the Navier-Stokes equations, or indeed other reductions of them.  For this reason, and in order to give the dual space results physical meaning, Cullen and Feldman \cite{feldman} mapped the solutions back to the original, physical coordinates and extended the results of \cite{benamou} and \cite{gangbo}, proving existence of weak Lagrangian solutions in physical space, in the incompressible case.  
We mention that very recently results on the existence of Eulerian solutions have been proved in the case of a two-dimensional torus, and of a convex open subset in 3-dimensional space \cite{ambnew}, \cite{ambnew2}.

\smallskip
In this paper, we make use of recent results in the analysis of ODEs in spaces of measures, in particular those of \cite{ambgangbo}, \cite{ambbook}, in order to provide an alternative proof of the dual space result in \cite{maroofi}.  We also extend considerably  the results of \cite{maroofi} to prove the existence of weak Lagrangian solutions of the fully compressible semi-geostrophic system in the original physical coordinates. 
 As in the incompressible case studied in \cite{feldman}, the proof is based on the existence of an appropriate flow map with rather low regularity; however we also show that, if we could assume additional regularity, then the solutions derived would determine classical solutions.

The paper is organised as follows: in Section \ref{notation} we introduce various definitions, and the notation to be used throughout; in Section \ref{key existing} we summarise the existing results regarding existence of solutions in dual space that are necessary for the theory presented in the following sections; in Section \ref{alternative} we use the results of  \cite{ambgangbo} to provide an alternative proof of the existence of solutions in dual space; in Section \ref{aims} we define the concept of a weak Lagrangian solution and  formulate our main theorem; in Section \ref{dual section} we summarise the existence proof and important properties of a Lagrangian flow in dual space, given in \cite{maroofi}; finally,  in Section \ref{physical section} we map the dual space Lagrangian flow to physical space and use it to prove our main result, namely  the existence of a weak Lagrangian solution.

\section{Useful Conventions, Notation and Definitions}\label{notation}
\setcounter{equation}{0}

We start by assembling all notation and conventions used in the sequel.

 \begin{equation}\label{list}
{\bf Physical\;quantities \;and \;constants}\end{equation}

\begin{enumerate}[label=(\textit{\roman{*}}), ref=\textit{(\roman{*})}]
	\item $\Omega $ denotes an open bounded convex set in $\mathbb{R} ^3 $, representing the physical domain containing the fluid; $\tau>0$ is a fixed positive constant; all functions in physical coordinates are defined for $(t, \mathbf x )\in [0,\tau)\times \Omega $;
	\item $\mathbf u(t, \mathbf x)$ represents the 3-D velocity of the fluid;
	\item $\mathbf u^g(t, \mathbf x) = (u_1^g(t, \mathbf x), u_2^g(t, \mathbf x), 0)$ represents the geostrophic velocity;
	\item $p(t, \mathbf x)$ represents the pressure;
	\item $\rho (t, \mathbf x)$ represents the density;
	\item $\theta (t, \mathbf x)$ represents the potential temperature. Given its physical meaning, we assume $\theta (t, \mathbf x)$  to be strictly positive and bounded;
		\item $\phi (\mathbf x)$ is the given geopotential representing gravitation and centrifugal forces.  We assume that $\phi\in C^2(\bar \Omega)$  and that $\frac{\partial } {\partial x_3}\phi (\mathbf x) \neq 0$ for all $\mathbf x \in \overline{\Omega}$;
			\item $f_\textrm{\scriptsize{cor}}$ denotes the Coriolis parameter, which we assume to be constant; indeed we will normalise this parameter to be equal to 1;
		\item $p_\textrm{\scriptsize{ref}}$ is the reference value of the pressure;
		\item $c_v$ is a constant representing the specific heat at a constant pressure;
	\item $c_p$ is a constant representing the specific heat at a constant volume;	
	\item $R$ represents the gas constant and satisfies $R = c_p - c_v$;
	\item $\kappa = c_p/c_v$ denotes the ratio of specific heats (this is approximately 1.4 for air).
\end{enumerate}

 \begin{equation}\label{list2}
{\bf Notations \; and \; other \; conventions}\end{equation}
 
\begin{itemize}
	

\item Throughout, we will only consider measures that are absolutely continuous with respect to Lebesgue measure.  Given an open set $A$ in $\mathbb{R}^3$, we will denote by
	
	 - \quad $P_{ac}(A)$  - the set of probability measures in $\mathbb{R}^3$ with supports contained in $A$;	 
	 
	 - \quad $\chi _A$ - the  characteristic function of $A$.
	\item Unless otherwise specified, measurable means Lebesgue measurable and $a.e.$ means Lebesgue-$a.e.$
	\item $D_t$ denotes the Lagrangian derivative, defined as $D_t = \partial _t + \mathbf u \cdot \nabla $, where ${\bf u}$ denotes the velocity of the flow.  
	\item  $\mathbf e_3$. denotes the unit vector $\mathbf  (0, 0, 1)$ in $\mathbb R^3$ 
	\item For convenience, we will sometimes use the notation $F_{(t)}(\cdot ) = F(t, \cdot )$ to denote the map $F$ evaluated at fixed time $t$.  
	\end{itemize}



 \begin{equation*}\label{list3}
{\bf Important  \; definitions}\end{equation*}

\begin{definition}\label{pac}
 We define
\[ P_{ac}^2(\mathbb{R}^3):=\{ \mu \in P_{ac}(\mathbb{R}^3) \textrm{ }: \textrm{ } \int_{\mathbb{R}^3 } \! |\mathbf x|^2 \, \mu(\mathbf x) d\mathbf x < + \infty  \} ,\] 
with tangent space 
\begin{equation}\label{tangent}
T_{\mu }P_{ac}^2(\mathbb{R}^3) = \overline{ \{ \nabla \varphi : \varphi \in C_c^\infty (\mathbb{R}^3) \} }^{L^2(\mu ; \mathbb{R}^3)}.
\end{equation}
\end{definition}


\begin{definition}\label{6}
Given two Borel probability densities $\mu _1(\cdot )$ and $\mu _2(\cdot )$ in $\mathbb{R}^3$, we define the \emph{Wasserstein-$2$ distance}, $W_2$, between $\mu _1$ and $\mu _2$ as follows:
\begin{equation}\label{wasserstein}
W^2_2(\mu _1, \mu _2) :=  \inf _{\gamma \in \Gamma (\mu _1, \mu _2)} \int_{\mathbb{R}^3 \times \mathbb{R}^3 } \! |\mathbf x - \mathbf y|^2 \, d\gamma (\mathbf x, \mathbf y) . 
\end{equation}
with
\[ \Gamma (\mu _1, \mu _2) = \{ \gamma \textrm{ } : \textrm{ } \gamma \textrm{ probability measure on }\mathbb{R}^3 \times \mathbb{R}^3 \textrm{ with marginals }\mu _1 \textrm{ and } \mu _2\}.\]
We denote by $\Gamma _0(\mu _1, \mu _2 )$ the set of minimisers of (\ref{wasserstein}).
\end{definition}
\noindent The Wasserstein distance indeed defines a distance in the space of  probability measures on $\mathbb R^3$ (or more generally, any complete separable metric space); it can be used as an optimal transport cost between these measures, see \cite{villanioldnew}.
\begin{definition}\label{def2}
Let $F : \mathbb{R}^3 \rightarrow \mathbb{R}^3$ be a measurable map.
 Let $\mu _1$ be a measure in $\mathbb{R}^3$.  Then we say that the measure $\mu _2 $ on $\mathbb{R}^3$ is the \emph{push forward} of $\mu _1$ with respect to $F$, denoted $F\textrm{\#}\mu _1 = \mu _2$, if
\[ \mu _2 [A] = \mu _1[F^{-1}(A)]\]
for all measurable $A \subset \mathbb{R}^3$.
If, in addition, $\mu _1 \in L^1(\mathbb{R}^3)$, $\mu _2 \in L^1(\mathbb{R}^3)$ and $F\textrm{\#}\mu _1 = \mu _2$, then (see  \cite[Corollary A.3]{feldman})
\begin{equation}\label{appendix}
\int_{\mathbb{R}^3} \! \varphi (F (\mathbf x)) \mu _1(\mathbf x) \, d \mathbf x = \int_{\mathbb{R}^3} \! \varphi (\mathbf X) \mu _2(\mathbf X)\, d\mathbf X, \quad  \forall\; \varphi \in L^\infty (\mathbb{R}^3).
\end{equation}
\end{definition}

\begin{remark}\label{abscont}
Since we are only interested in measures that are absolutely continuous with respect to Lebesgue measure, we can guarantee existence of a unique minimiser $\gamma _0$ in (\ref{wasserstein}); see, for example, \cite[Section 6.2.3]{ambbook}.  We also have that $\gamma _0 = (\textit{\textbf{id}}, \mathbf R^{\mu _2} _{\mu _1} )\textrm{\#}\mu _1$ for some map $\mathbf R_{\mu _1}^{\mu _2} :\mathbb{R}^3 \rightarrow \mathbb{R}^3$ which coincides $\mu _1-a.e.$ with the gradient of a convex function.  Thus, the map $\mathbf R_{\mu _1}^{\mu _2} $ is the unique minimiser of
\begin{equation}
 \mathbf R \rightarrow \int_{\mathbb{R}^3 } \! \left |\mathbf x - \mathbf R(\mathbf x)\right |^2 \, \mu _1(\mathbf x)d \mathbf x \label{Wassmin}\end{equation}
 over all $\mathbf R$ such that $\mathbf R \textrm{\#}\mu _1 = \mu _2$.
 \end{remark}


%




\section{The semi-geostrophic equations: formulation and existing results}\label{key existing}
\setcounter{equation}{0}

The fully compressible semi-geostrophic equations posed in the domain $[0,\tau)\times\Omega$,  with rigid boundaries $\partial \Omega$,  are the following system of equations (see, for example, \cite{maroofi}):
\begin{eqnarray}
&&\label{commom}D_t \mathbf u^g + f_\textrm{\scriptsize{cor}} \mathbf e _3 \times \mathbf u + \nabla \phi + \frac{1}{\rho } \nabla p = 0,\\
&&\label{comad}D_t \theta = 0,\\
&&\label{comcont}D_t \frac{1}{\rho } = \frac{1}{\rho }\nabla \cdot \mathbf u,\hspace{128pt} (t,\mathbf x)\in [0, \tau )\times \Omega\\
&&\label{comgeo}f_\textrm{\scriptsize{cor}} \mathbf e _3 \times \mathbf u^g + \nabla \phi + \frac{1}{\rho } \nabla p = 0,\\
&&\label{comstate}
p = R\rho \theta \left (\frac{p}{p_\textrm{\scriptsize{ref}}}\right )^{\frac{\kappa -1}{\kappa }};\\
&&\label{combound}
\mathbf u \cdot \mathbf n = 0,\hspace{150pt} (t,\mathbf x)\in[0, \tau ) \times \partial \Omega.
\end{eqnarray}
The unknowns in the above equations are $\mathbf u^g = (u^g_1, u^g_2, 0)$, $\mathbf u = (u_1, u_2, u_3)$, $p$, $\rho $, $\theta $.

\smallskip
Equation (\ref{commom}) is the momentum equation; (\ref{comad}) represents the adiabatic assumption; (\ref{comcont}) is the continuity equation and (\ref{comgeo}) represents hydrostatic and geostrophic balance.  The equation  (\ref{comstate}) is the equation of state which relates the thermodynamic quantities to each other, and (\ref{combound}) is the rigid boundary condition, where $\mathbf n$ is the outward normal to $\partial \Omega $.

The semi-geostrophic equations are a valid approximation  to the compressible Euler equations when
$
UL<<1,
$
and are accurate when 
$
\frac{H}{L} < \frac{1}{N},
$
where $U$ is a typical scale for horizontal speed; $L$ is a typical horizontal scale; $H$ is a typical vertical scale; $N$ is the buoyancy frequency.

The energy associated with the flow, known as the \emph{geostrophic energy}, is defined as
\begin{equation}\label{geoenergy}
E(t)=\int_{\Omega } \! \left [ \frac{1}{2}\left |\mathbf u^g \right |^2(t,\mathbf x) + \phi(\mathbf x) + c_v\theta(t,\mathbf x) \left (\frac{p(t,\mathbf x)}{p_\textrm{\scriptsize{ref}}}\right )^{\frac{\kappa -1}{\kappa }} \right ] \rho(t,\mathbf x) \, d \mathbf x.
\end{equation}

\smallskip
In what follows, we set $f_\textrm{\scriptsize{cor}}=1$.

\subsection{Dual formulation}
In \cite{maroofi}, solutions were obtained using a transformation into the so-called  {\em dual (geostrophic) coordinates} $\mathbf y = (y_1, y_2, y_3)$.  The coordinate transformation  is given by:
\[\mathbf T:\Omega\to\Lambda\subset\mathbb R^3,\qquad  \mathbf T(t, \mathbf x) = (T_1(t, \mathbf x), T_2(t, \mathbf x), T_3(t, \mathbf x)) = (y_1, y_2, y_3), \]
with
\begin{equation}\label{transform}
y_1 = x_1 + u_2^g(t, \mathbf x),\qquad y_2 = x_2 - u_1^g(t, \mathbf x), \qquad y_3 = \theta (t, \mathbf x).
\end{equation}
Note that, by (\ref{list})(\emph{vi}),  $\overline{\Lambda } \subset \mathbb R^2\times[\delta,\frac 1 \delta]$ for some $0<\delta<1$.

Using this transformation, as well as (\ref{comstate}), it was shown in \cite{maroofi} that we can write the energy  in (\ref{geoenergy})   as
\begin{equation}\label{energy}
\begin{split}
E (t) &= E(t,\sigma,\mathbf T)= \int_\Omega \! \frac{\frac{1}{2}\{ |x_1 - y_1|^2 + |x_2 - y_2|^2\} + \phi (\mathbf x)}{y_3}\sigma(t,\mathbf x) \, d\mathbf x + K_1\int_\Omega \! (\sigma(t,\mathbf x))^\kappa \, d\mathbf x \\
&=  \int_\Omega \! \frac{\frac{1}{2}\{ |x_1 - T_1(t,\mathbf x)|^2 + |x_2 - T_2(t,\mathbf x)|^2\} + \phi (\mathbf x)}{T_3(t,\mathbf x)}\sigma(t,\mathbf x) \, d\mathbf x + K_1\int_\Omega \! (\sigma(t,\mathbf x))^\kappa \, d\mathbf x 
\end{split}
\end{equation}
where $K_1 := c_v\left (\frac{R}{p_\textrm{\scriptsize{ref}}}\right )^{\kappa -1}$ is constant and
\begin{equation}\label{sigma}
\sigma (t,\mathbf x):= \theta (t,\mathbf x)\rho(t,\mathbf x).
\end{equation}
Using the results of \cite{shutts}, it was shown that stable solutions of (\ref{commom})-(\ref{combound}) correspond to solutions that, at each fixed time $t$,  minimise the energy $E$ given by (\ref{energy}). 
We refer to this requirement as \emph{Cullen's stability condition}.  

This stability condition can be formulated in terms of optimal transport concepts.  Indeed, define the potential density $\nu := \mathbf T \textrm{\#}\sigma $ as the push forward of the measure $\sigma $ under the map $\mathbf T$.  For a stable solution, the energy in (\ref{energy}) can be written in the following form
\begin{equation}\label{energy2}
E_\nu (\sigma ) = {\cal E}(\sigma , \nu ) + K_1\int_\Omega \! (\sigma(t,\mathbf x))^\kappa \, d\mathbf x,
\end{equation}
where
\begin{equation}\label{mkp}
{\cal E}(\sigma , \nu ) = \inf _{\overline{\mathbf T} \textrm{\#} \sigma = \nu } \int_\Omega \! c(\mathbf x, \overline{\mathbf T} (t,\mathbf x)) \sigma (t,\mathbf x)  \, d\mathbf x
\end{equation}
with $c$ a cost function, given by
\begin{equation}\label{cost}
c(\mathbf x, \mathbf y) = \frac{\frac{1}{2}\{ |x_1 - y_1|^2 + |x_2 - y_2|^2\} + \phi (\mathbf x)}{y_3}.
\end{equation}
Hence the energy minimisation required by Cullen's stability condition can be reformulated as an optimal transport problem.  Indeed, this  condition is equivalent to the condition  that, at each fixed time, the pair $(\sigma , \mathbf T)$ minimises the energy (\ref{energy}) amongst all pairs $(\mu , \overline{\mathbf T})$ with $\mu \in P_{ac}(\Omega )$ and $\overline{\mathbf T} \textrm{\#} \mu = \nu $.  In other words, Cullen's stability condition amounts to the requirement that the change of variables $\mathbf T(t, \cdot )$ from physical coordinates to geostrophic coordinates given by (\ref{transform}) is the optimal transport map between $\sigma$ and $\nu$.  

In \cite{maroofi}, it is shown that, given a potential density $\nu $, there exists  a unique minimiser $\sigma$ of (\ref{energy2}). Given this pair $\nu$, and $\sigma$,  there always exists a unique minimiser $\mathbf T$ in (\ref{mkp}), given by the optimal map in the transport of $\sigma $ to $\nu $; the map $\mathbf T$ admits a unique inverse $\mathbf T^{-1}$, which is the optimal map in the transport of $\nu $ to $\sigma $ with cost $\tilde{c}(\mathbf y, \mathbf x) = c(\mathbf x, \mathbf y)$.

\smallskip
Using the stability condition and the transformation given by (\ref{transform}), the compressible semi-geostrophic equations (\ref{commom})-(\ref{comgeo}) can be written in dual variables as follows \cite{maroofi}:

\begin{doublespace}
\begin{eqnarray}
\label{comdual1}&&\partial _t \nu (t, \mathbf y)+ \nabla \cdot (\nu (t, \mathbf y)\mathbf w(t, \mathbf y)) = 0 \qquad \textrm{in} \quad [0, \tau ) \times \Lambda,\\
\label{comdual2}&&\mathbf w(t, \mathbf y)= \mathbf u^g (t, \mathbf S(t, \mathbf y)) = \mathbf e_3 \times \left [\mathbf y - \mathbf S(t, \mathbf y)\right ],\\
\label{comdual3}&&\mathbf S(t, \cdot ) = \mathbf T^{-1} (t, \cdot ), \\ 
\label{comdual4}&&\mathbf T(t, \cdot ) \textrm{ is the unique optimal transport map in (\ref{mkp})},\\
\label{comdual5}&&\sigma (t, \cdot ) \textrm{ minimises } E_{\nu(t, \cdot )} (\cdot ) \textrm{ over } P_{ac} (\Omega ),\\
\label{comdual6}&&\nu (0, \cdot ) = \nu _0 (\cdot ) \in L^r(\Lambda_0), \, r \in (1, \infty ),\;\Lambda_0\subset \mathbb R^3\;{\rm compact}.
\end{eqnarray}
\end{doublespace}
Equation (\ref{comdual1}) is the continuity equation satisfied by the potential density; (\ref{comdual2}) defines the geostrophic velocity in dual variables; (\ref{comdual3}), (\ref{comdual4}) and (\ref{comdual5}) are required for a stable solution; and (\ref{comdual6}) is the prescribed initial condition.

The main result of \cite{maroofi} is the existence of a stable weak solution of the semi-geostrophic system in dual variables (\ref{comdual1})-(\ref{comdual6}), proved by an approximation procedure. In the same paper, several important results concerning optimal transport are proved.  
In Section \ref{min}, we summarise these results.

\begin{remark}\label{support}
Since $\nu _0$ has compact support,  for all $t\in(0,\tau)$ the solution $\nu(t,\mathbf y)$ of the evolution (\ref{comdual1}) starting at  time $t=0$ from $\nu_0$ has compact support in $\mathbb R^3$. By (\ref{list})(\emph{vi}),  $supp(\nu)$ is contained in a bounded open set $\Lambda $, which is dependent on $\tau $, such that $\overline{\Lambda} \subset \mathbb{R}^2 \times [\delta , \frac{1}{\delta}]$, for some $\delta$ with $0<\delta<1$. This follows from a standard fixed-point argument; see, for example, \cite{cullen}, \cite{loeper}.
\end{remark}

\begin{remark}\label{fariaremark}
In \cite{faria}, Faria \emph{et al.} have extended the results of \cite{feldman} for the incompressible equations to the case of an  initial potential density $\nu _0$ in $L^1$.  
\end{remark}

\subsection{Existence and uniqueness for the energy minimisation}\label{min}

Assume that
\begin{equation}\label{25}
\begin{split}
& \bullet \qquad \Omega \textrm{ and } \Lambda \textrm{ are bounded open domains in } \mathbb{R}^3,\\
& \bullet \qquad  \textrm{ } \Omega \textrm{ is convex };\quad \overline{\Lambda} \subset \mathbb{R}^2 \times \left [\delta , \frac{1}{\delta }\right ] \textrm{ for some } \delta \,\textrm{with} \,0<\delta<1,
\end{split}
\end{equation}
where the assumption on the vertical coordinate of $\Lambda $ is justified by (\ref{list})(\emph{vi}).
Cullen and Maroofi proved in \cite[Theorem 4.1]{maroofi} that, given $\nu \in P_{ac}(\Lambda )$, there exists a unique $\sigma \in P_{ac}(\Omega )$ that minimises the energy $E_\nu (\cdot )$:
\begin{theorem}\label{4.1}
Assume that $\Omega $ and $\Lambda $ satisfy (\ref{25}).  Let $\nu \in P_{ac}(\Lambda )$ and assume that $\phi $ satisfies (\ref{list})(\emph{ix}).  Then there exists a unique minimiser $\sigma $ of $E_\nu (\cdot )$, over $P_{ac}(\Omega )$.  This minimiser is given by (\ref{energy2}).
\end{theorem}
%

By fixing this minimising $\sigma $, we can complete the energy minimisation by considering (\ref{mkp}) as an optimal transport problem.  Indeed, given two probability densities $\sigma \in P_{ac}(\Omega )$ and $\nu \in P_{ac}(\Lambda )$, consider the following optimal transport problem:  
\begin{equation}\label{22}
\inf _{\overline{\mathbf T} \textrm{\#}\sigma = \nu }I_\sigma [\overline{\mathbf T}], \quad \textrm{ with } I_\sigma [\overline{\mathbf T}] := \int_{\Omega } \! c(\mathbf x, \overline{\mathbf T}(\mathbf x))\sigma (\mathbf x) \, d \mathbf x,
\end{equation}
\[ \]
where $c(\mathbf x, \mathbf y)$, given by (\ref{cost}), is the cost of transporting one unit of mass from $\mathbf x$ to $\mathbf y$.

The Kantorovich relaxation of (\ref{22}) amounts to finding $\gamma \in \Gamma (\sigma , \nu )$ that minimises
\begin{equation}\label{27}
\overline{I}[\gamma ] := \int_{\Omega \times \Lambda } \! c(\mathbf x, \mathbf y) \, d \gamma (\mathbf x, \mathbf y),
\end{equation}
with  $\Gamma (\sigma , \nu)$ given in Definition \ref{6}.  Note that every transport map $\overline{\mathbf T}$ generates a transport plan $\gamma _{\overline{\mathbf T}} \in \Gamma (\sigma , \nu )$ defined by
\[ \gamma _{\overline{\mathbf T}} := (\textit{\textbf{id}}, \overline{\mathbf T})\textrm{\#}\sigma , \]
such that $I _{\sigma }[\overline{\mathbf T}] = \overline{I}[\gamma _{\overline{\mathbf T}}]$.

Now consider the Kantorovich dual problem corresponding to (\ref{22}) (see, for example, \cite[Chapter 5]{villanioldnew}):
\begin{equation}\label{26}
\sup _{(f, g) \in Lip_c}J_{(\sigma , \nu )}(f, g),\qquad \textrm{ with } J_{(\sigma , \nu )}(f, g) := \int_{\Omega } \! f(\mathbf x)\sigma (\mathbf x) \, d \mathbf x + \int_{\Lambda } \! g(\mathbf y)\nu (\mathbf y) \, d \mathbf y,
\end{equation}
where
\[ Lip_c := \{ (f, g) \textrm{ } : \textrm{ } f \in W^{1, \infty }(\Omega ),\textrm{ } g \in W^{1, \infty }(\Lambda ), \textrm{ } f(\mathbf x) + g(\mathbf y) \leqslant c(\mathbf x, \mathbf y) \textrm{ for all } (\mathbf x, \mathbf y) \in \Omega \times \Lambda \}.\]
The solution of the dual problem is a crucial ingredient of proving the main result of this section, namely the existence result for the optimal transport problem.   See \cite{maroofi} for all relevant definitions and for a proof.
\begin{theorem}\label{3.3}
Assume that $\Omega $ and $\Lambda $ satisfy (\ref{25}).  Let $\sigma \in P_{ac}(\Omega )$ and $\nu \in P_{ac}(\Lambda )$.  Assume that $\phi $ satisfies (\ref{list})(\emph{ix}).  Then there exist maps $\mathbf T : \Omega \rightarrow \Lambda $ and $\mathbf S : \Lambda \rightarrow \Omega $, unique $\sigma -a.e.$ and $\nu -a.e.$ respectively, such that
\begin{enumerate}[label=(\textit{\roman{*}}), ref=\textit{(\roman{*})}]
	\item $\mathbf T$ is optimal in the transport of $\sigma $ to $\nu $ with cost $c(\mathbf x, \mathbf y)$,
	\item $\mathbf S$ is optimal in the transport of $\nu $ to $\sigma $ with cost $\tilde{c}(\mathbf y, \mathbf x) = c(\mathbf x, \mathbf y)$,
\item $\mathbf S$ and $\mathbf T$ are inverses, i.e. $\mathbf S\circ \mathbf T(\mathbf x) = \mathbf x$ for $\sigma -a.e.$ $\mathbf x$ and $\mathbf T\circ \mathbf S(\mathbf y) = \mathbf y$ for $\nu -a.e.$ $\mathbf y$,
\item $\gamma _0 = (\textbf{id}, \mathbf T) \textrm{\#} \sigma $ is a minimiser of the relaxed optimal transport problem (\ref{27}),
\item For $J_{(\sigma , \nu )}(f, g)$ defined by (\ref{26}),  the following equality holds:
\[ \sup _{(f, g) \in Lip_c}J_{(\sigma , \nu )}(f, g) = \inf _{\gamma \in \Gamma (\sigma , \nu )} \overline{I}[\gamma ] = \inf _{\overline{\mathbf T} \textrm{\#}\sigma = \nu }I_\sigma [\overline{\mathbf T}].\]
\end{enumerate}
\end{theorem}

\section{The main existence result in dual space - an alternative proof}\label{alternative}
\setcounter{equation}{0}

\subsection{Statement of the theorem}

\begin{theorem}\label{5.5}

Let $1 < r < \infty $ and $\nu _0 \in L^r (\Lambda _0 )$ be an initial potential density with support in $\Lambda _0$, where $\Lambda _0 $ is a bounded open set in $\mathbb{R} ^3$ with $\overline{\Lambda _0} \subset \mathbb{R} ^2 \times [\tilde{\delta }, 1/\tilde{\delta }]$ for some $0 < \tilde{\delta } < 1$.  
 Let $\Omega $ be an open bounded convex set in $\mathbb{R} ^3$.  Assume that $c(\cdot ,\cdot )$ is given by (\ref{cost}) and that $\phi $ satisfies (\ref{list})(\emph{ix}).  Then the system of semi-geostrophic equations in dual variables (\ref{comdual1})-(\ref{comdual6}) has a stable weak solution $(\sigma, \mathbf T)$ such that, with $\nu (t, \cdot ) = \mathbf T(t, \cdot ) \textrm{\#} \sigma (t, \cdot )$ and $\mathbf w$ as in (\ref{comdual2}),
\begin{enumerate}[label=(\textit{\roman{*}}), ref=\textit{(\roman{*})}]
	\item \[ \nu (\cdot , \cdot ) \in L^r ((0, \tau ) \times \Lambda ), \qquad \left\| \nu (t, \cdot )\right\| _{L^r (\Lambda )} \leqslant \left\| \nu _0 (\cdot )\right\| _{L^r (\Lambda )}, \qquad \forall\; t \in [0, \tau ],
\]
\item  \[ \sigma (t, \cdot ) \in W^{1, \infty } (\Omega ), \qquad \left\| \sigma (t, \cdot )\right\| _{W^{1, \infty } (\Omega )} \leqslant C = C(\Omega, \Lambda, c(\cdot , \cdot ), \kappa , K_1 ),\qquad \forall\;t \in [0, \tau ], \]
\item \[ \left\| \mathbf w(t, \cdot ) \right\| _{L^\infty (\Lambda )} \leqslant C = C(\Omega , \Lambda)\qquad \forall\;t \in [0, \tau ],  \]
\end{enumerate}
where $\Lambda $ is a bounded open domain in $\mathbb{R}^3$ containing $supp (\nu )$, such that $\overline{\Lambda} \subset \mathbb{R}^2 \times \left [\delta , \frac{1}{\delta }\right ]$ for some $0<\delta <1$.
\end{theorem}

 The proof of this theorem is given in  \cite[Theorem 5.5]{maroofi}, using a time-approximation argument, similar in spirit to the original argument of the proof given by Benamou and Brenier  in \cite{benamou} of the existence of solution of the incompressible system in a fixed domain.
 
 The result of this section is an alternative proof of this theorem. The proof given here avoids the time-discretisation and makes use of the general theory developed in \cite{ambgangbo}  on the solution of Hamiltonian ODEs, showing explicitly how it may be applied to the dual space existence problem. 

\subsection{The solution of Hamiltonian ODEs}

  We briefly summarise the  general results of  \cite{ambgangbo} on the solution of Hamiltonian ODEs.

\begin{remark}
In what follows we deal with concave rather than convex functions.  Hence we replace all definitions in \cite{ambgangbo} regarding subdifferentiability and $(\lambda )-$convexity with the following definitions regarding superdifferentiability and $(\lambda )-$concavity.  This replacement does not affect the results of \cite{ambgangbo}.
\end{remark}



\begin{definition}\label{superdifferential}
Let $H: P_{ac}^2(\mathbb{R}^3) \rightarrow (-\infty , +\infty ]$ be a proper, upper semicontinuous function and let $\nu \in D(H)$.  We say that $\mathbf v \in L^2(\nu ; \mathbb{R}^3)$ belongs to the \textit{Fr\'{e}chet superdifferential} $\partial H(\nu )$ if
\[ H(\nu _h ) \leqslant H(\nu ) + \int_{\mathbb{R}^3  } \!  \mathbf v(\mathbf y) \cdot (\mathbf R_{\nu }^{\nu _h}(\mathbf y) - \mathbf y )\, \nu (\mathbf y) \, d\mathbf y + o(W_2(\nu , \nu _h))\]
as $\nu _h \rightarrow \nu $.  We denote by $\partial _0 H(\nu )$ the element of $\partial H(\nu )$ of minimal $L^2(\nu ; \mathbb{R}^3)-$norm.  Note that, by the minimality of its norm, $\partial _0 H(\nu )$ belongs to $ \partial H(\nu ) \cap T_{\nu }P_{ac}^2(\mathbb{R}^3)$.
\end{definition}

In the following lemma we state a continuity property of optimal plans or maps.
\begin{lemma}\label{amb3.3}
Assume that $\{\nu _n\} _{n = 1}^\infty $, $\{\mu _n\} _{n = 1}^\infty $ are bounded sequences in $P_{ac}^2(\mathbb{R}^3)$ narrowly converging to $\nu $ and $\mu $ respectively.  Assume that $\Gamma _0(\nu , \mu )$ contains a unique plan $\gamma _0$ induced by the optimal map $\mathbf R_\nu ^\mu : \mathbb{R}^3 \rightarrow \mathbb{R}^3$ (see Remark \ref{abscont}).  Then
\[ \lim _{n \rightarrow +\infty } \int_{\mathbb{R}^3 } \! g(\mathbf y, \mathbf R_{\nu _n}^{\mu _n}(\mathbf y))\nu _n(\mathbf y) \, d \mathbf y = \int_{\mathbb{R}^3 } \! g(\mathbf y, \mathbf R_\nu ^\mu(\mathbf y))\nu (\mathbf y) \, d \mathbf y, \]
where $\mathbf R_{\nu _n}^{\mu _n}$ is optimal in the transport of $\nu _n$ to $\mu _n$, and for any continuous function $g : \mathbb{R}^3 \times \mathbb{R}^3  \rightarrow \mathbb{R}^3 $ satisfying
\[ \lim _{|(\mathbf y, \mathbf x)| \rightarrow +\infty } \frac{|g|(\mathbf y, \mathbf x)}{|\mathbf y|^2 + |\mathbf x|^2} = 0.\]


Assume furthermore that there exists a closed ball $B_r$, of finite radius $r$, containing the supports of $\mu _n$ and $\mu $.  Then there exist Lipschitz, convex functions $u_n, u:\mathbb{R}^3 \rightarrow \mathbb{R} \cup \{ + \infty \}$ such that $\nabla u_n = \mathbf R_{\nu _n}^{\mu _n}$ $\nu _n-a.e.$ in $\mathbb{R}^3$ and $\nabla u = \mathbf R_{\nu }^{\mu }$ $\nu -a.e.$ in $\mathbb{R}^3$.  In addition, there exists a subsequence $\{ n_k \} ^\infty _{k=1}$ of integers such that
\begin{equation}\label{amb21}
\nabla u_{n_k} \rightarrow \nabla u \qquad a.e. \textrm{ in } \mathbb{R}^3.
\end{equation}
\end{lemma}
\begin{proof}
See \cite[Proposition 7.1.3]{ambbook} for proof.
\end{proof}
%
%

\begin{definition}\label{lambdaconcave}
Let $H: P_{ac}^2(\mathbb{R}^3) \rightarrow (-\infty , +\infty ]$ be proper and let $\lambda \in \mathbb{R}$.  Let $\pi _1: \mathbb{R}^3 \times \mathbb{R}^3 : (\mathbf x, \mathbf y) \rightarrow \mathbf x$ and $\pi _2: \mathbb{R}^3 \times \mathbb{R}^3: (\mathbf x, \mathbf y) \rightarrow \mathbf y$ be the first and second projections of $\mathbb{R}^3 \times \mathbb{R}^3$ onto $\mathbb{R}^3$.  We say that $H$ is \textit{$\lambda -$concave} if for every $\nu _1$, $\nu _2 \in P_{ac}^2(\mathbb{R}^3)$ and every optimal transport plan $\gamma \in \Gamma _0(\nu _1, \nu _2)$ we have
\[ H(\nu _{(t)}) \geqslant (1 - t)H(\nu _1) + tH(\nu _2) - \frac{\lambda }{2}t(1 - t)W_2^2(\nu _1, \nu _2) \]
for all $t \in [0, 1]$, where $\nu _{(t)} = ((1-t)\pi _1 + t\pi _2)\textrm{\#}\gamma $.
\end{definition}


\begin{proposition}\label{characterise}
Let $H: P_{ac}^2(\mathbb{R}^3) \rightarrow (-\infty , +\infty ]$ be upper semicontinuous and $\lambda -$concave for some $\lambda \in \mathbb{R}$ and let $\nu \in D(H)$.  Then, the following condition is equivalent to $\mathbf v \in \partial H(\nu )$: 
\[ H(\nu _h ) \leqslant H(\nu ) + \int_{\mathbb{R}^3  } \!  \mathbf v(\mathbf y) \cdot ( \mathbf R_{\nu }^{\nu _h}(\mathbf y) - \mathbf y)\nu (\mathbf y) \, d \mathbf y + \frac{\lambda }{2}W_2^2(\nu , \nu _h)\]
for all $\nu _h \in P_{ac}^2(\mathbb{R}^3)$.
\end{proposition}
\begin{proof}
See \cite[Proposition 4.2]{ambgangbo} for proof.
\end{proof}

Following \cite{ambgangbo}, we define Hamiltonian ODE's as follows:

\begin{definition}\label{hamode}
Let $H: P_{ac}^2(\mathbb{R}^3) \rightarrow (-\infty , +\infty ]$ be a proper, upper semicontinuous function.  Define the linear transformation $\tilde{J} : \mathbb{R}^3 \rightarrow \mathbb{R}^3$ by 
\begin{equation}\label{tildeJ}
\tilde{J}(v_1(\mathbf y), v_2(\mathbf y), v_3(\mathbf y)) = y_3(-v_2(\mathbf y), v_1(\mathbf y), 0),
\end{equation}
for all $\mathbf v(\mathbf y) \in \mathbb{R}^3$.  We say that an absolutely continuous curve $\nu _{(t)} : [0, \tau ] \rightarrow D(H)$ is a \textit{Hamiltonian ODE} relative to $H$, starting from $\nu _0 \in P_{ac}^2(\mathbb{R}^3)$, if there exist $\mathbf v_{(t)} \in L^2(\nu _{(t)};\mathbb{R}^3)$ with $\left\| \mathbf v_{(t)} \right\| _{L^2(\nu _{(t)})} \in L^1(0, \tau )$, such that
\begin{equation}\label{33} \begin{cases}
\frac{d}{dt}\nu _{(t)} + \nabla \cdot (\tilde{J} \mathbf v_{(t)} \nu _{(t)}) = 0, \qquad \nu _{(0)} = \nu _0, \qquad t \in (0, \tau )\\
\mathbf v_{(t)} \in T_{\nu _{(t)}}P_{ac}^2(\mathbb{R}^3) \cap \partial H(\nu _{(t)}) \qquad \textrm{for }a.e.\textrm{ } t.
\end{cases}\end{equation}
\end{definition}
The main result of \cite{ambgangbo} concerns Hamiltonians $H$ satisfying the following properties:\\*[3mm]
(H1) \textit{ There exist constants } $C_0 \in (0, +\infty )$, $R_0 \in (0, +\infty ]$ \textit{such that for all }$\nu \in P_{ac}^2(\mathbb{R}^3)$ \textit{with} $W_2(\nu , \nu _0) < R_0$ \textit{ we have } $\nu \in D(H)$, $\partial H(\nu ) \neq \emptyset $ \textit{ and } $\mathbf v = \partial _0 H(\nu )$ \textit{ satisfies } $|\mathbf v(\mathbf z)| \leqslant C_0(1 + |\mathbf z|)$ \textit{ for } $\nu -a.e.$ $\mathbf z \in \mathbb{R}^3$.\\*[3mm]
(H2) \textit{ If } $\nu ,\, \nu _n \in P_{ac}^2(\mathbb{R}^3)$, $\sup _n W_2(\nu _n, \nu_0)<R_0$ \textit{ and } $\nu _n \rightarrow \nu $ \textit{ narrowly, then there exists a subsequence } $n(k)$ \textit{ and functions } $\mathbf v_k$, $\mathbf v$ \textit{ such that } $\mathbf v_k = \partial _0 H(\nu _{n(k)})$ $\nu _{n(k)}-a.e.$, $\mathbf v = \partial _0 H(\nu )$ $\nu -a.e.$ \textit{ and } $\mathbf v_k \rightarrow \mathbf v$ $a.e.$ \textit{ in } $\mathbb{R}^3$ \textit{ as } $k \rightarrow +\infty $.\\*[3mm]
To ensure the constancy of $H$ along the solutions of the Hamiltonian system we consider also:\\*[3mm]
(H3) $H:P_{ac}^2(\mathbb{R}^3) \rightarrow (-\infty , +\infty ]$ \textit{ is proper, upper semicontinuous and }$\lambda -$\textit{concave for some }$\lambda \in \mathbb{R}$.\\*[3mm]

For Hamiltonians $H$ as above, the following  result holds (see \cite[Theorem 6.6]{ambgangbo} for full details and proof):
\begin{theorem}\label{6.6}
Assume that (H1) and (H2) hold for $H(\nu )$ and that $\tau > 0$ satisfies
\begin{equation}\label{maxtime}
C_0 \tau \sqrt{24(1 + e^{(25C_0^2 + 1)\tau }(1 + M_2(\nu _0)))}<R_0.
\end{equation}
Then there exists an absolutely continuous Hamiltonian flow $\nu _{(t)} \in P_{ac}^2(\mathbb{R}^3 )$,  $\nu_{(t)}:[0, \tau] \rightarrow D(H)$
starting from 
$\nu_0 \in P_{ac}^2(\mathbb{R}^3)$, satisfying (\ref{33}), 
such that the velocity field $\mathbf v_{(t)}$ coincides with $\partial _0 H(\nu _{(t)})$ for $a.e.$ $t \in [0, \tau ]$.  Furthermore, the function $t \rightarrow \nu _{(t)}$ is Lipschitz continuous.
Finally, there exists a function $l(r)$ depending only on $\tau $ and $C_0$ such that
\begin{equation}\label{bound1} \nu _0 \geqslant m_r \, \, a.e. \textrm{ on }B_r \textrm{ for all }r>0 \quad \implies \quad \nu _{(t)} \geqslant m_{l(r)} \, \, a.e. \textrm{ on }B_r\textrm{ for all }r>0
\end{equation}
and
\begin{equation}\label{bound2} \nu _0 \leqslant M_r \, \, a.e. \textrm{ on }B_r \textrm{ for all }r>0 \quad \implies \quad \nu _{(t)} \leqslant M_{l(r)} \, \, a.e. \textrm{ on }B_r\textrm{ for all }r>0.
\end{equation}
If in addition (H3) holds, then $t \mapsto H(\nu _{(t)})$ is constant.
\end{theorem}

\subsection{An alternative proof of Theorem \ref{5.5}}

We now apply Theorem \ref{6.6} to our problem in order to obtain directly the existence result of Theorem \ref{5.5}.  

Let $\Lambda $, $\Omega $ be as in (\ref{25}).  For $\nu \in P_{ac}(\Lambda )$, define the Hamiltonian $H$ by
\begin{equation}\label{ourH}
H(\nu )= \inf _{\mu \in P_{ac}(\Omega )} \left \{ {\cal E}(\nu , \mu ) + K_1\int_\Omega \! (\mu (\mathbf x))^\kappa \, d\mathbf x \right \} ,
\end{equation}
where 
\begin{equation} {\cal E}(\nu , \mu ) = \inf _{\overline{\mathbf S} \textrm{\#} \nu = \mu } \int_\Lambda \! \tilde{c}(\mathbf y, \overline{\mathbf S} (\mathbf y)) \nu ( \mathbf y)  \, d\mathbf y \label{ourE}\end{equation}
with $\tilde{c}(\mathbf y, \mathbf x) = c(\mathbf x, \mathbf y)$ defined by (\ref{cost}).  
We begin with the following:
\begin{proposition}\label{Hok}
Let $\Omega $ and $\Lambda $ satisfy (\ref{25}). 
 Let the Hamiltonian $H(\nu)$ on $P_{ac}(\Lambda )$ be defined by (\ref{ourH}). Then $H$ is superdifferentiable, upper semicontinuous and $(-2)-$concave.
\end{proposition}


\begin{proof}

Given $\nu \in P_{ac}(\Lambda )$, denote by $\sigma $ the minimiser in (\ref{ourH}).  The existence and uniqueness of this minimiser follows from Theorem \ref{4.1}.  For any $\nu _h \in P_{ac}(\Lambda)$ we have
\begin{eqnarray*}
H(\nu _h) &&= \inf _{\mu \in P_{ac}(\Omega )} \left \{ {\cal E}(\nu _h, \mu ) + K_1\int_\Omega \! (\mu (\mathbf x))^\kappa \, d\mathbf x \right \} \\
&& \leqslant {\cal E}(\nu _h, \sigma ) + K_1\int_\Omega \! (\sigma (\mathbf x))^\kappa \, d\mathbf x.
\end{eqnarray*}
First, recall that we can guarantee that there exists a unique optimal transport map $\mathbf R^{\nu _h}_{\nu }$ from $\nu $ to $\nu _h$ with respect to the Wasserstein cost function $d(\mathbf y, \mathbf y_h) = \frac{1}{2}|\mathbf y - \mathbf y_h|^2$; see Remark \ref{abscont}. 


We consider now transport with respect to the cost function $\tilde{c}(\mathbf y, \mathbf x) = c(\mathbf x, \mathbf y)$ given by (\ref{cost}).  Let $\mathbf S_\nu ^\sigma $ be the optimal map in the transport of $\nu $ to $\sigma $ and let $\mathbf S_{\nu_h}^\sigma $ be the optimal map in the transport of $\nu _h$ to $\sigma $.  Therefore, we have
\[ \inf _{\mathbf S \textrm{\#} \nu = \sigma } \int_\Lambda \! \tilde{c}(\mathbf y, \mathbf S (\mathbf y)) \nu (\mathbf y)  \, d\mathbf y =  \int_\Lambda \! \tilde{c}(\mathbf y, \mathbf S_\nu ^\sigma(\mathbf y)) \nu (\mathbf y)  \, d\mathbf y \]
and
\[\inf _{\mathbf S \textrm{\#} \nu _h= \sigma } \int_\Lambda \! \tilde{c}(\mathbf y, \mathbf S (\mathbf y)) \nu _h(\mathbf y)  \, d\mathbf y = \int_\Lambda \! \tilde{c}(\mathbf y, \mathbf S_{\nu _h}^\sigma (\mathbf y)) \nu _h(\mathbf y)  \, d\mathbf y .\]
The existence of $\mathbf S_{\nu }^\sigma$ and $\mathbf S_{\nu _h}^\sigma $ follows from Theorem \ref{3.3}.  Note that, since $(\mathbf S_{\nu }^\sigma \circ (\mathbf R^{\nu _h}_{\nu })^{-1}) \textrm{\#} \nu _h= \sigma $ and since $\mathbf S_{\nu_h}^\sigma $ is optimal in the transport of $\nu _h$ to $\sigma $, we have
\[ \int_\Lambda \! \tilde{c}(\mathbf y, \mathbf S_{\nu _h}^\sigma (\mathbf y)) \nu _h(\mathbf y)  \, d\mathbf y \leqslant \int_\Lambda \! \tilde{c}(\mathbf y, \mathbf S_{\nu }^\sigma \circ (\mathbf R^{\nu _h}_{\nu })^{-1} (\mathbf y)) \nu _h(\mathbf y)  \, d\mathbf y. \]
It follows that
\begin{align*}
H(\nu _h) - H(\nu ) &\leqslant {\cal E}(\nu _h, \sigma ) + K_1\int_\Omega \! (\sigma (\mathbf x))^\kappa \, d\mathbf x - {\cal E}(\nu , \sigma ) - K_1\int_\Omega \! (\sigma (\mathbf x))^\kappa \, d\mathbf x \\
&= \int_\Lambda \! \tilde{c}(\mathbf y, \mathbf S_{\nu _h}^\sigma (\mathbf y)) \nu _h(\mathbf y)  \, d\mathbf y - \int_\Lambda \! \tilde{c}(\mathbf y, \mathbf S_{\nu }^\sigma (\mathbf y)) \nu (\mathbf y)  \, d\mathbf y \\
&\leqslant \int_\Lambda \! \tilde{c}(\mathbf y, \mathbf S_{\nu }^\sigma \circ (\mathbf R^{\nu _h}_{\nu })^{-1} (\mathbf y)) \nu _h(\mathbf y)  \, d\mathbf y - \int_\Lambda \! \tilde{c}(\mathbf y, \mathbf S_{\nu }^\sigma (\mathbf y)) \nu (\mathbf y)  \, d\mathbf y \\
&= \int_\Lambda \! \tilde{c}(\mathbf R^{\nu _h}_{\nu }(\mathbf y), \mathbf S_{\nu }^\sigma (\mathbf y)) \nu (\mathbf y)  \, d\mathbf y - \int_\Lambda \! \tilde{c}(\mathbf y, \mathbf S_{\nu }^\sigma (\mathbf y)) \nu (\mathbf y)  \, d\mathbf y \\
&= \int_\Lambda \! \bigg[ \tilde{c}(\mathbf R^{\nu _h}_{\nu }(\mathbf y), \mathbf S_{\nu }^\sigma (\mathbf y)) - \tilde{c}(\mathbf y, \mathbf S_{\nu }^\sigma (\mathbf y)) \bigg]\nu (\mathbf y)  \, d\mathbf y, \\
& = \int_\Lambda \!\nabla \tilde{c}(\mathbf y, \mathbf S_{\nu }^\sigma (\mathbf y))\cdot [\mathbf R^{\nu _h}_{\nu }(\mathbf y) - \mathbf y] \nu (\mathbf y)  \, d\mathbf y + o(W_2(\nu , \nu _h)). 
\end{align*}
\begin{equation}\label{otherfloaty}
\end{equation}
where we have used (\ref{appendix}).

Hence, using Definition \ref{superdifferential}, we conclude that $\nabla \tilde{c}(\mathbf y, \mathbf S_{\nu }^\sigma (\mathbf y)) \in \partial H(\nu )$. 
 Thus, $\partial H(\nu )$ is non-empty, $H$ is superdifferentiable and we can use \cite[Proposition 10.12]{villanioldnew} to conclude that $H$ is semi-concave, i.e. 
\begin{equation}\label{semiconcave}
H \textrm{ is  } (-2)-\textrm{concave}.
\end{equation}
Also, from the narrow continuity of ${\cal E}(\cdot, \cdot )$ (see \cite[Theorem 3.4]{maroofi}) and the uniform convergence of $\sigma $ as the minimiser of (\ref{ourH}) (see \cite[Lemma 4.3]{maroofi}), we have that 
\begin{equation}\label{usc}
H \textrm{ is upper semicontinuous.}
\end{equation} 
From (\ref{semiconcave}) and (\ref{usc}), we have that (H3) holds.
\end{proof}

The following proposition yields a proof of Theorem \ref{5.5}, alternative to the proof given in  \cite{maroofi}.
\begin{proposition}\label{mainalt}
Let $1 < r < \infty $ and $\nu _0 \in L^r (\Lambda _0 )$ be an initial potential density with support in $\Lambda _0$, where $\Lambda _0 $ is a bounded open set in $\mathbb{R} ^3$ with $\overline{\Lambda _0} \subset \mathbb{R} ^2 \times [\tilde{\delta }, 1/\tilde{\delta }]$ for some $0 < \tilde{\delta } < 1$.  Let $\Omega $ and $\Lambda $ satisfy (\ref{25}) and 
 let the Hamiltonian $H$ be defined by (\ref{ourH}).  Then, there exists a Hamiltonian flow $\nu _{(t)} \in P_{ac}(\Lambda )$ and constant $\tau >0$ 
 such that
\[ \frac{d}{dt}\nu _{(t)} + \nabla \cdot (\tilde{J}( \mathbf v_{(t)}) \nu _{(t)}) = 0, \qquad \nu _{(0)} = \nu _0, \qquad t \in (0, \tau )\]
where $\tilde{J}( \mathbf v_{(t)})=\mathbf w$ a.e. in $[0,\tau ]$.
\end{proposition}

\begin{proof}


We compute $\partial _0H(\nu )$ (as defined in Definition \ref{superdifferential}) explicitly to show that the conditions required to apply Theorem \ref{6.6} hold.  From the definition of $\tilde{J}$ in (\ref{tildeJ}), velocity fields transporting $\nu $ will have vanishing components in the $y_3$ direction so that we need only consider variations of $\nu $ in the $(y_1, y_2)-$directions.  Thus, to characterise the elements of $\partial H(\nu )$, we let $\tilde{\varphi} \in C_c^\infty (\mathbb{R}^2)$ 
 and define $\varphi (y_1,y_2,y_3):=\tilde{\varphi}(y_1,y_2)$ for all $\mathbf y \in \mathbb{R}^3$.  We then set


  
 \[ \mathbf g_s(\mathbf y) = ((g_s)_1(\mathbf y), (g_s)_2(\mathbf y), (g_s)_3(\mathbf y)) = \mathbf {y} + s\nabla \varphi (\mathbf y).\]
Note that $(g_s)_3(\mathbf y) = y_3$ and, for $|s|$ sufficiently small, $\mathbf g_s$ is the gradient of a convex function, since $\mathbf g_s(\mathbf y) = \nabla (\frac{1}{2}\mathbf y^2 + s\varphi )$.   Define $\nu _s = \mathbf g_s \textrm{\#} \nu $.   Denote by $\sigma _s$ the minimiser in
 \[ H(\nu _s)= \inf _{\mu \in P_{ac}(\Omega )} \left \{ {\cal E}(\nu _s, \mu ) + K_1\int_\Omega \! (\mu (\mathbf x))^\kappa \, d\mathbf x \right \} .\]
 The existence and uniqueness of the minimiser $\sigma _s$ follows from Theorem \ref{4.1}.  Let $\xi \in \partial H(\nu )$.  
 Combining the $(-2)-$concavity of $H$ and (\ref{usc}) with Proposition \ref{characterise}, we obtain
 \begin{equation}\label{inequality}
  H(\nu _s) - H(\nu ) - \int_\Lambda \!  \xi (\mathbf y) \cdot (\mathbf R^{\nu _s}_{\nu }(\mathbf y) - \mathbf {y} )  \nu (\mathbf y) \, d\mathbf y + W_2^2(\nu , \nu _s) \leqslant 0.
  \end{equation}
Brenier's polar factorisation theorem states that any suitable mapping from $\nu $ to $\nu _s$ can be uniquely factorised as the composition of a measure-preserving mapping and the gradient of a convex function (see, for example, \cite[Chapter 3]{villani}).  Since, for $|s|$ sufficiently small, $\mathbf g_s$ is the gradient of a convex function, we conclude that
 \begin{equation*}
 W_2^2(\nu , \nu _s) = \int_\Lambda \! |\mathbf{y} - \mathbf R^{\nu _s}_{\nu }(\mathbf y)|^2 \nu (\mathbf y) \, d\mathbf y = \int_\Lambda \! |\mathbf{y} - \mathbf g_s(\mathbf y)|^2 \nu (\mathbf y) \, d\mathbf y = s^2\int_\Lambda \! |\nabla \varphi (\mathbf y)|^2 \nu (\mathbf y) \, d\mathbf y 
 \end{equation*}
 and
 \[ \int_\Lambda \!  \xi (\mathbf y) \cdot ( \mathbf R^{\nu _s}_{\nu } (\mathbf y)- \mathbf {y} )  \nu (\mathbf y) \, d\mathbf y = \int_\Lambda \!  \xi (\mathbf y) \cdot ( \mathbf g_s(\mathbf y) - \mathbf{y} )  \nu (\mathbf y) \, d\mathbf y = s\int_\Lambda \! \xi (\mathbf y) \cdot \nabla \varphi (\mathbf y)  \nu (\mathbf y) \, d\mathbf y .\]
Combining this with (\ref{inequality}), we therefore obtain
 \begin{align*}
 -s\int_\Lambda \!  \xi (\mathbf y) \cdot \nabla \varphi (\mathbf y)  \nu (\mathbf y) \, d\mathbf y & + s^2\int_\Lambda \! |\nabla \varphi (\mathbf y)|^2\nu (\mathbf y) \, d\mathbf y \leqslant H(\nu ) - H(\nu _s)\\
 &\leqslant {\cal E}(\nu , \sigma _s) - {\cal E}(\nu _s, \sigma _s) \\
 &= \int_\Lambda \! \tilde{c}(\mathbf y, \mathbf S _\nu ^{\sigma _s}(\mathbf y)) \nu (\mathbf y)  \, d\mathbf y - \int_\Lambda \! \tilde{c}(\mathbf y, \mathbf S _{\nu _s} ^{\sigma _s} (\mathbf y)) \nu _s(\mathbf y)  \, d\mathbf y \\
& \leqslant \int_\Lambda \! \tilde{c}(\mathbf y, \mathbf S _{\nu _s}^{\sigma _s}\circ \mathbf g_s(\mathbf y)) \nu (\mathbf y)  \, d\mathbf y - \int_\Lambda \! \tilde{c}(\mathbf y, \mathbf S _{\nu _s} ^{\sigma _s} (\mathbf y)) \nu _s(\mathbf y)  \, d\mathbf y 
\end{align*}
\begin{equation}
 = \int_\Lambda \! \tilde{c}(\mathbf g_s^{-1}(\mathbf y), \mathbf S _{\nu _s}^{\sigma _s}(\mathbf y)) \nu _s(\mathbf y)  \, d\mathbf y - \int_\Lambda \! \tilde{c}(\mathbf y, \mathbf S _{\nu _s} ^{\sigma _s} (\mathbf y)) \nu _s(\mathbf y)  \, d\mathbf y ,
\label{floaty}
\end{equation}
since $\mathbf g_s \textrm{\#} \nu = \nu _s$.  Here $\mathbf S _{\nu }^{\sigma _s}$ denotes the optimal transport map from $\nu $ to $\sigma _s$ and $\mathbf S _{\nu _s}^{\sigma _s}$  denotes the optimal transport map from $\nu _s$ to $\sigma _s$ with respect to the cost function $\tilde{c}(\cdot , \cdot )$.  The existence of $\mathbf S _{\nu }^{\sigma _s}$ and $\mathbf S _{\nu _s}^{\sigma _s}$ follows from Theorem \ref{3.3}.

Note that 
 \[ \mathbf g_s^{-1}(\mathbf y) = \mathbf y - s\nabla \varphi (\mathbf y) + \frac{s^2}{2}\nabla ^2\varphi (\mathbf y)\nabla \varphi (\mathbf y) + \epsilon (s, \mathbf y),\]
where $\epsilon $ is a function such that $|\epsilon (s, \mathbf y) | \leqslant |s|^3\| \varphi \| _{C^3(\mathbb{R}^3)}$.  

Combining this expression for $\mathbf g_s^{-1}$ with (\ref{floaty}) and using $\frac{\partial }{\partial y_3}\varphi = 0$, we conclude that

\begin{align*}
-s\int_\Lambda \! & \xi (\mathbf y) \cdot \nabla \varphi (\mathbf y)  \nu (\mathbf y) \, d\mathbf y  + s^2\int_\Lambda \! |\nabla \varphi (\mathbf y)|^2\nu (\mathbf y) \, d\mathbf y \\
&\leqslant \int_\Lambda \! \left [ \tilde{c}(\mathbf g_s^{-1}(\mathbf y), \mathbf S _{\nu _s}^{\sigma _s}(\mathbf y)) - \tilde{c}(\mathbf y, \mathbf S _{\nu _s} ^{\sigma _s} (\mathbf y)) \right ]\nu _s(\mathbf y)  \, d\mathbf y\\
& = \int_\Lambda \! \left [\frac{\frac{1}{2} \left \{ \left | (g_s)_ 1^{-1}(\mathbf y) - (S_{\nu _s} ^{\sigma _s})_1 (\mathbf y)\right | ^2 + \left | (g_s)_2^{-1}(\mathbf y) - (S_{\nu _s} ^{\sigma _s})_2 (\mathbf y)\right | ^2 \right \} + \phi (\mathbf S _{\nu _s} ^{\sigma _s} (\mathbf y))}{y_3} \right. \\ &\qquad - \left. \frac{\frac{1}{2} \left \{ \left | y_1 - (S_{\nu _s} ^{\sigma _s})_1 (\mathbf y)\right | ^2 + \left | y_2 - (S_{\nu _s} ^{\sigma _s})_2 (\mathbf y)\right | ^2 \right \} + \phi (\mathbf S _{\nu _s} ^{\sigma _s} (\mathbf y))}{y_3}\right ] \nu _s(\mathbf y)\, d\mathbf y\\
& = \int_\Lambda \! \frac{1}{y_3}\left [\frac{1}{2}\left \{ \left | y_1 - s\frac{\partial }{\partial y_1}\varphi (\mathbf y) - (S_{\nu _s}^{\sigma _s})_1 (\mathbf y)\right | ^2 + \left | y_2 - s\frac{\partial }{\partial y_2}\varphi (\mathbf y)- (S_{\nu _s}^{\sigma _s})_2 (\mathbf y)\right | ^2\right \} \right. \\ & \qquad - \left. \frac{1}{2}\left \{ \left |y_1 - (S_{\nu _s} ^{\sigma _s})_1 (\mathbf y) \right | ^2 + \left | y_2 - (S_{\nu _s} ^{\sigma _s})_2 (\mathbf y)\right | ^2\right \} \right ] \nu _s(\mathbf y)\, d\mathbf y + o(s)\\
&= s \int_\Lambda \! \left (\frac{\mathbf S_{\nu _s} ^{\sigma _s}(\mathbf y) - \mathbf y}{y_3}\right ) \cdot \nabla \varphi (\mathbf y) \, \nu _s (\mathbf y)\, d\mathbf y + o(s).
\end{align*}
Recall now that $\nu _s \rightarrow \nu $ in $P_{ac}(\Lambda )$ and $\sigma _s \rightarrow \sigma $ in $P_{ac}(\Omega )$ as $s \rightarrow 0$, hence Lemma \ref{amb3.3} gives
\[ -s\int_\Lambda \! \xi (\mathbf y) \cdot \nabla \varphi (\mathbf y)  \nu (\mathbf y) \, d\mathbf y  + s^2\int_\Lambda \! |\nabla \varphi (\mathbf y)|^2\nu (\mathbf y) \, d\mathbf y \leqslant s \int_\Lambda \! \left ( \frac{\mathbf S_{\nu } ^{\sigma }(\mathbf y) - \mathbf y}{y_3} \right ) \cdot \nabla \varphi (\mathbf y) \, \nu  (\mathbf y)\, d\mathbf y + o(s).\]
Dividing both sides first by $s>0$, then by $s<0$ and letting $|s| \rightarrow 0$ we obtain
\[ -\int_\Lambda \!  \xi (\mathbf y) \cdot \nabla \varphi (\mathbf y) \,  \nu (\mathbf y) \,  d\mathbf y=  \int_\Lambda \! \left ( \frac{ \mathbf S_{\nu } ^{\sigma }(\mathbf y) - \mathbf y}{y_3} \right ) \cdot \nabla \varphi (\mathbf y)  \, \nu  (\mathbf y)\, d\mathbf y.\]
Thus, we have that $\tilde{J}(\pi _\nu \xi (\mathbf y)) = \tilde{J}\left (\left ( \frac{  \mathbf y - \mathbf S_{\nu } ^{\sigma }(\mathbf y)}{y_3} \right )\right )$, where $\pi _\nu : L^2(\nu ; \Lambda ) \rightarrow T_\nu P_{ac}^2(\Lambda )$ denotes the canonical orthogonal projection
, where the tangent space is defined in (\ref{tangent}).  The minimality of the norm of $\partial _0 H$ then gives
\begin{equation}\label{oursupdiff}
\tilde{J}(\partial _0 H(\nu )) = \tilde{J}\left (\left ( \frac{  \mathbf y - \mathbf S_{\nu } ^{\sigma }(\mathbf y) }{y_3} \right ) \right ) = \mathbf w(\mathbf y),
\end{equation}
where $\mathbf w$ is defined as in (\ref{comdual2}).

We can now check directly that conditions (H1) and (H2) hold.  Using Theorem \ref{3.3}, we may conclude from (\ref{oursupdiff}) that
\[\tilde{J}(\partial _0 H(\nu )) =  \mathbf w(\mathbf y) = \tilde{J}(\nabla g_0(\mathbf y)). \]  
  Condition (H1) then follows from $g_0 \in W^{1, \infty }(\Lambda )$, where $g_0$ is the solution of the dual problem (\ref{26}).  Condition (H2) follows from the 
 stability of optimal maps (see, for example, \cite[Corollary 5.23]{villanioldnew}).  
Hence we may apply 
 the result of Theorem \ref{6.6} to conclude that there exists a Hamiltonian flow $\nu _{(t)}$ such that
\[ \frac{d}{dt}\nu _{(t)} + \nabla \cdot (\tilde{J}( \mathbf v_{(t)}) \nu _{(t)}) = 0, \qquad \nu _{(0)} = \nu _0, \qquad t \in (0, \tau )\]
where $\tilde{J} (\mathbf v_{(t)}) = \tilde{J}( \partial _0 H(\nu _{(t)}) )$ for $a.e.$ $t \in [0, \tau ]$.  By (\ref{oursupdiff}), this then completes the proof that the dual space continuity equation (\ref{comdual1}), with velocity field defined as in (\ref{comdual2}), is satisfied.  In addition, from (H3) and the definition of $\tilde{J}$, the energy associated with the flow is conserved.

\end{proof}

\begin{proof}[Proof of Theorem \ref{5.5}]
By the definition of $\mathbf w$ in Proposition \ref{mainalt}, we have that $(\sigma , \mathbf T)$ is a stable solution of (\ref{comdual1})-(\ref{comdual6}), where $\mathbf T = \mathbf S^{-1}$ (see Theorem \ref{3.3}).  Theorem \ref{5.5} (\emph{i}) follows from (\ref{bound1}), (\ref{bound2}); Theorem \ref{5.5} (\emph{ii}) follows from the fact that $\sigma $ is a minimiser of (\ref{ourH}) (see \cite[Theorem 4.2]{maroofi}); Theorem \ref{5.5} (\emph{iii}) follows from the definition of $\mathbf w$ in terms of $\nabla g_0$ and the fact that, as a solution of the dual problem (\ref{26}), $g_0 \in W^{1, \infty }(\Lambda )$.
\end{proof}

\section{Lagrangian statement of the equations in physical space}\label{aims}
\setcounter{equation}{0}
In this section, our  aim is to prove the existence of a weak Lagrangian solution of the fully compressible semi-geostrophic system (\ref{commom})-(\ref{combound}).  We also show that, with additional regularity, a weak Lagrangian solution would determine a weak Eulerian solution.  These statements are proved following closely the treatment given in \cite{feldman} of the incompressible case.

We begin by rewriting the system (\ref{commom})-(\ref{combound}) in a form that enables us to state the equations in Lagrangian form.  These new equations (with prescribed initial conditions) are

\begin{doublespace}
\begin{eqnarray}\label{momentum}
&&D_t\mathbf T(t, \mathbf x) =  \mathbf e _3 \times \left [\mathbf T(t, \mathbf x) - \mathbf x\right ],\\
&&\label{continuity}\partial _t\sigma (t, \mathbf x) + \nabla \cdot (\sigma (t, \mathbf x)\mathbf u(t, \mathbf x)) = 0,\\
&&\label{newboy}\nabla _{\mathbf x}c(\mathbf x, \mathbf T(t, \mathbf x))
 + \kappa K_1\nabla ((\sigma(t, \mathbf x))^{\kappa - 1})
 = 0,\\
&&\label{bound}\mathbf u \cdot \mathbf n = 0 \hspace{3mm} \textrm{ on } [0, \tau ) \times \partial \Omega ,\\
&&\label{initial} \sigma (0, \mathbf x) = \sigma _0(\mathbf x), \quad \mathbf T(0, \mathbf x) = \mathbf T_0, \quad \mathbf T_0 \textrm{\#}\sigma _0\in L^r(\Lambda_0), \;\Lambda_0\subset \mathbb R^3\;{\rm compact},
\end{eqnarray}
\end{doublespace}
\noindent where $r \in (1, \infty )$, $c(\mathbf x , \mathbf T(t, \mathbf x) )$ is defined in (\ref{cost}) and $K_1$ is defined as in (\ref{energy}).
\begin{proposition}
A solution of (\ref{momentum})-(\ref{initial}) determines a solution of the original system (\ref{commom})-(\ref{combound}).
\end{proposition}
\begin{proof}
Given a 
  solution $(\mathbf u , \mathbf T, \sigma )$ of (\ref{momentum})-(\ref{initial}), we set the function $\mathbf u$ in (\ref{commom})-(\ref{combound}) to be equal to the function $\mathbf u$ in (\ref{momentum})-(\ref{initial}) and we define
\begin{equation}\label{theta rho ug}
\theta (t, \mathbf x) := T_3(t, \mathbf x), \hspace{10mm} \rho (t, \mathbf x) := \frac{\sigma (t, \mathbf x)}{\theta (t, \mathbf x)}, \hspace{10mm} \mathbf u^g(t, \mathbf x):=  \mathbf e _3 \times[\mathbf T(t, \mathbf x) - \mathbf x].
\end{equation}
Then (\ref{comad}) follows from (\ref{theta rho ug}) and equation $D_tT_3 = 0$ of (\ref{momentum}).  Using this together with (\ref{theta rho ug}) and the fact that (\ref{continuity}) is satisfied, we see that (\ref{comcont}) holds.  


In order to show that (\ref{comgeo}) holds, we first rearrange (\ref{comstate}) to obtain $p = R^\kappa p_\textrm{\scriptsize{ref}}^{1 - \kappa }(\rho \theta )^\kappa $.  Thus,
\begin{equation}\label{gradp}
\nabla p = R^\kappa p_\textrm{\scriptsize{ref}}^{1 - \kappa }\kappa (\rho \theta )^{\kappa - 1}\nabla (\rho \theta ).
\end{equation}
Then, substituting (\ref{theta rho ug}) into (\ref{newboy}) gives 
\[  \mathbf e _3 \times \mathbf u^g + \nabla \phi + \theta \kappa K_1 \nabla (\rho \theta )^{\kappa -1} = 0.\]
Recalling that $K_1 = c_v\left (\frac{R}{p_\textrm{\scriptsize{ref}}}\right )^{\kappa -1}$, this becomes
\[  \mathbf e _3 \times \mathbf u^g + \nabla \phi + \theta \kappa c_v\left (\frac{R}{p_\textrm{\scriptsize{ref}}}\right )^{\kappa -1} \nabla (\rho \theta )^{\kappa -1} = 0, \]
and therefore
\[  \mathbf e _3 \times \mathbf u^g + \nabla \phi + \frac{1}{\rho } \kappa (\kappa - 1)c_v\left (\frac{R}{p_\textrm{\scriptsize{ref}}}\right )^{\kappa -1}  (\rho \theta )^{\kappa -1}\nabla (\rho \theta ) = 0. \]
Then, using (\ref{list})(\emph{xiv}), (\ref{list})(\emph{xv}) and (\ref{gradp}) we obtain
\[  \mathbf e _3 \times \mathbf u^g + \nabla \phi + \frac{1}{\rho } \nabla p = 0 \]
and thus conclude that (\ref{comgeo}) is satisfied.

Finally, we obtain (\ref{commom}) using (\ref{comgeo}) together with the first two components of (\ref{momentum}) and the definition of $\mathbf u^g$ in (\ref{theta rho ug}).  Hence, we obtain a solution of the original system of semi-geostrophic equations.
\end{proof}

We define a weak solution of (\ref{momentum})-(\ref{initial}) as follows:
\begin{definition}\label{2.1}
Let $\Omega $ be an open bounded convex set in $\mathbb{R}^3$.  Let $\mathbf u: [0, \tau) \times \Omega \rightarrow \mathbb{R} ^3 $ satisfy $\mathbf u \in L^1([0, \tau ) \times \Omega )$ and let $\mathbf T : [0, \tau) \times \Omega \rightarrow \mathbb{R} ^3$ satisfy $\mathbf T \in L^\infty ([0, \tau ) \times \Omega )$.  Also, let $\sigma (t, \cdot ) \in W^{1, \infty }(\Omega )$ for all $t \in [0, \tau ]$.  Then $(\mathbf u , \mathbf T, \sigma )$ is a \emph{weak (Eulerian) solution} of (\ref{momentum})-(\ref{initial}) if
\[ \int_{[0,\tau )\times \Omega } \! \{ \mathbf T (t, \mathbf x )\cdot [\partial_t \varphi (t, \mathbf x ) + (\mathbf u (t, \mathbf x ) \cdot \nabla )\varphi (t, \mathbf x )] + \mathbf e_3 \times [\mathbf T(t, \mathbf x) - \mathbf x]\cdot \varphi (t, \mathbf x ) \} \sigma (t,\mathbf x)\, dtd\mathbf x \]
\begin{equation} \label{eweak1}
+ \int_\Omega \! \mathbf T_0(\mathbf x)\cdot \varphi (0, \mathbf x)\sigma _0(\mathbf x) \, d\mathbf x = 0,
\end{equation}
for any $\varphi \in C_c^1([0, \tau ) \times \Omega )$, and
\begin{equation}\label{eweak2}
\int_{[0,\tau )\times \Omega } \! \{ \partial _t \psi (t, \mathbf x) + (\mathbf u (t, \mathbf x) \cdot \nabla ) \psi (t, \mathbf x)\} \sigma (t, \mathbf x) \, dtd\mathbf x + \int_\Omega \! \psi (0, \mathbf x) \sigma _0(\mathbf x) \, d\mathbf x = 0,
\end{equation}
for any $\psi \in C_c^1([0, \tau ) \times \Omega )$.
\end{definition}

Given a solution $\mathbf T(t, \mathbf x)$ and $\mathbf x(t, \mathbf y) = \mathbf S(t, \mathbf y)$ of the system (\ref{comdual1})-(\ref{comdual6}) in dual coordinates, formally we have that $(\mathbf u, \mathbf T, \sigma )$ satisfy (\ref{momentum})-(\ref{initial}).  However, due to the low regularity of $\mathbf T(t, \mathbf x)$ and thus $\mathbf S(t, \mathbf y)$ obtained as a weak solution of (\ref{comdual1})-(\ref{comdual6}), this argument is not rigorous.  Indeed, the regularity obtained in \cite{maroofi} may not be enough to provide a solution of the Eulerian problem.


Instead, we seek to find weak Lagrangian solutions of the problem (\ref{momentum})-(\ref{initial}).  Here we define such solutions and then we use the methods of \cite{feldman} to prove their existence.  We define the Lagrangian flow map
\[ \mathbf F : [0, \tau ) \times \Omega \rightarrow \Omega \]
corresponding to the velocity $\mathbf u$.  Note that $\mathbf F$ maps $\Omega $ to itself so that the boundary conditions of the problem are respected.  Then, we can rewrite the system (\ref{momentum})-(\ref{initial}) in terms of this Lagrangian flow $\mathbf F$ and define the corresponding weak solution $(\mathbf F, \mathbf T, \sigma )$.

\begin{definition}\label{2.4}
Let $\Omega $ be an open bounded convex set in $\mathbb{R}^3$ and let $\tau > 0$.  Let $\mathbf T \in L^\infty ([0, \tau ) \times \Omega )$ and let $\sigma (t, \cdot ) \in W^{1, \infty }(\Omega )$ for all $t \in [0, \tau ]$.  For $1 \leqslant q < \infty $, let $\mathbf F: [0, \tau) \times \Omega \rightarrow \Omega $ be a Borel map, satisfying
\begin{equation} \label{Fcts}
\mathbf F \in C([0,\tau ); L^q (\Omega ; \mathbb{R} ^3)).
\end{equation}
\[ \]
Then $(\mathbf F, \mathbf T, \sigma )$ is called a \emph{weak Lagrangian solution} of (\ref{momentum})-(\ref{initial}) in $[0, \tau) \times \Omega $ if:
\begin{enumerate}[label=(\textit{\roman{*}}), ref=\textit{(\roman{*})}]
	\item For $\sigma _0-a.e.$ $\mathbf x \in \Omega $,
 \[ \mathbf F(0, \mathbf x) = \mathbf x, \hspace{10mm} \mathbf T(0, \mathbf x) = \mathbf T_0(\mathbf x). \]
\item For any $t>0$, the mapping
\[ \mathbf F_{(t)} = \mathbf F(t, \cdot ): [0, \tau ) \times \Omega \rightarrow \Omega \]
pushes forward the probability measure $\sigma _0(\cdot ) $ to $\sigma (t, \cdot )$, i.e.
\begin{equation} \label{pushforward}
 \mathbf F_{(t)} \textrm{\#} \sigma _0(\cdot )  = \sigma (t, \cdot ). 
 \end{equation}
\item There exists a Borel map
\[ \mathbf F^*: [0, \tau ) \times \Omega \rightarrow \Omega \]
such that, for every $t \in (0, \tau )$, the map 
\[ \mathbf F_{(t)}^* = \mathbf F^*(t, \cdot):\Omega \rightarrow \Omega \]
satisfies
\[ \mathbf F_{(t)}^* \circ \mathbf F_{(t)}(\mathbf x) = \mathbf x \textrm{ and } \mathbf F_{(t)} \circ \mathbf F_{(t)}^*(\mathbf x) = \mathbf x, \]
for $\sigma -a.e.$ $\mathbf x \in \Omega $, and 
\[ \mathbf F_{(t)}^* \textrm{\#} \sigma (t, \cdot )  = \sigma _0(\cdot ). \]
\[ \]
\item The function
\begin{equation} \label{Z}
\mathbf Z(t, \mathbf x) := \mathbf T(t, \mathbf F_{(t)}(\mathbf x))
\end{equation}
is a weak solution of
\begin{equation} \label{Zequation}
\begin{split}
&\partial _t \mathbf Z(t, \mathbf x) = \mathbf e_3 \times \left [\mathbf Z(t, \mathbf x) - \mathbf F_{(t)}( \mathbf x)\right ] \hspace{3mm} \textrm{ in } [0, \tau ) \times \Omega ,\\
&\mathbf Z(0, \mathbf x) = \mathbf T_0(\mathbf x) \hspace{3mm} \textrm{ in } \Omega ,
\end{split}
\end{equation}
in the following sense:\\
for any $\varphi \in C_c^1([0,\tau ) \times \Omega ; \mathbb{R} ^3)$,
\begin{equation} \label{Zweak}
\begin{split}
\int_{[0,\tau )\times \Omega } \! \{ \mathbf Z(t, \mathbf x) \cdot \partial _t &\varphi (t, \mathbf x) + \mathbf e_3 \times [\mathbf Z(t, \mathbf x) - \mathbf F_{(t)}( \mathbf x)] \cdot \varphi (t, \mathbf x)\}\sigma _0(\mathbf x) \, dtd\mathbf x \\
&+ \int_\Omega \! [\mathbf T_0(\mathbf x)\cdot \varphi (0, \mathbf x)]\sigma _0(\mathbf x) \, d\mathbf x = 0.
\end{split}
\end{equation}
\end{enumerate}
\end{definition}
\begin{remark}
Let us comment on Definition \ref{2.4}:
\begin{itemize}
  \item Continuity in time of $\mathbf F$ considered as a map on $[0, \tau )$ with values in $L^q(\Omega)$ as required in (\ref{Fcts}), combined with the initial condition for $\mathbf F$ in (\emph{i}), implies that
  \begin{equation}\label{2.39}
  \lim _{t \rightarrow 0+} \left\| \mathbf F_{(t)} - \textit{\textbf{id}} \right\| _{L^q(\Omega )} = 0.
  \end{equation}
Furthermore, the continuity property (\ref{Fcts}) may be interpreted as continuity of particle paths in physical space.  
	\item Property (\emph{ii}) is the Lagrangian form of the mass conservation principle in equation (\ref{continuity}) with boundary condition (\ref{bound}).
	\item Equation (\ref{Zequation}) is the Lagrangian form of (\ref{momentum}).
		\item Equation (\ref{Zweak}) is derived by multiplying (\ref{Zequation}) by $\sigma _0(\mathbf x)$ as well as the test function and then integrating by parts, as in \cite[Definition 3.4]{feldman}.  The reason for this choice 
		 will become clear in the proof of Proposition \ref{2.6}.
		\item We have omitted equation (\ref{newboy}) since this holds as a result of the energy minimisation; see \cite[Theorem 4.2]{maroofi}.
\end{itemize}
\end{remark}

In order to justify Definition \ref{2.4}, we must show that a weak Lagrangian solution corresponds to a weak (Eulerian) solution of (\ref{momentum})-(\ref{initial}), as defined by Definition \ref{2.1}.  Indeed, we prove that, with additional regularity property $\partial _t \mathbf F \in L^\infty ([0, \tau )\times \Omega )$, a weak Lagrangian solution $(\mathbf F, \mathbf T, \sigma )$ as defined in Definition \ref{2.4} determines a weak (Eulerian) solution of (\ref{momentum})-(\ref{initial}) and, furthermore, that a smooth Lagrangian solution determines a classical solution of (\ref{momentum})-(\ref{initial}).  This is the content of the following result: 

\begin{proposition}\label{2.6}      
Let $\Omega $ be an open bounded convex set in $\mathbb{R}^3$ and let $\tau >0$.  Let $(\mathbf F , \mathbf T, \sigma )$ be a weak Lagrangian solution of (\ref{momentum})-(\ref{initial}) in $[0, \tau ) \times \Omega $.
\begin{enumerate}[label=(\textit{\roman{*}}), ref=\textit{(\roman{*})}]
\item If $\partial _t \mathbf F \in L^\infty ([0, \tau )\times \Omega ; \mathbb{R}^3)$, then the function
\begin{equation}\label{u}
\mathbf u (t,\mathbf x ) := (\partial _t \mathbf F)(t, \mathbf F_{(t)}^*(\mathbf x ))
\end{equation}
satisfies $\mathbf u \in L^\infty ([0, \tau )\times \Omega ; \mathbb{R}^3)$ and $(\mathbf u, \mathbf T, \sigma )$ is a weak Eulerian solution of (\ref{momentum})-(\ref{initial}) in $[0, \tau ) \times \Omega $ in the sense of Definition \ref{2.1}.
\item If $(\mathbf F, \mathbf F^* ,\mathbf T) \in C^2 ([0, \tau ] \times \overline{\Omega })$, then the function (\ref{u}) satisfies $\mathbf u \in C^1 ([0, \tau ] \times \overline{\Omega }; \mathbb{R}^3)$, and $(\mathbf u , \mathbf T, \sigma )$ is a classical solution of (\ref{momentum})-(\ref{initial}) in $[0, \tau ) \times \Omega $.
\end{enumerate}
\end{proposition}
\begin{proof}
Let us first prove (\emph{i}).  Since $\mathbf F^*$ is a Borel map and, by our additional regularity assumption, $\partial _t\mathbf F \in L^\infty ([0, \tau ) \times \Omega )$, we know that the right-hand side of (\ref{u}) is a bounded measurable function.  Therefore, $\mathbf u \in L^\infty ([0, \tau ) \times \Omega )$.  Now, in order to prove that $(\mathbf u, \mathbf T, \sigma )$ is a weak Eulerian solution of (\ref{momentum})-(\ref{initial}), we must show that (\ref{eweak1}) and (\ref{eweak2}) hold.  We begin with (\ref{eweak2}).  Let $\psi \in C_c^1([0, \tau ) \times \overline{\Omega})$, so that the support of $\psi $ in $t$ is a closed subset of $[0, \tau )$, and fix $t \in (0, \tau )$.  Note that, since $\sigma (t, \cdot ) \in W^{1, \infty }(\Omega )$ for all $t \in [0, \tau ]$ (by Theorem \ref{5.5} (\emph{ii})) we have that $\sigma (t, \cdot ) \in L^1 ([0, \tau ) \times \Omega )$ and, since $\mathbf F_{(t)} \textrm{\#} \sigma _0(\cdot )  = \sigma (t, \cdot )$, we can apply (\ref{appendix}) to yield
\[ \int_\Omega \! (\partial _t\psi )(t, \mathbf F_{(t)}(\mathbf x))\sigma _0(\mathbf x) \, d\mathbf x = \int_\Omega \! \partial _t\psi (t, \mathbf X)\sigma (t, \mathbf X) \, d\mathbf X.\]
Then, applying the chain rule, integrating both sides with respect to $t$ and using our assumption that $\partial _t\mathbf F \in L^\infty ([0, \tau ) \times \Omega )$, we obtain
\begin{align*}
 \int_{[0, \tau ) \times \Omega } \! \{ \partial _t[\psi (t, \mathbf F&_{(t)}(\mathbf x))] - \partial _t\mathbf F_{(t)}(\mathbf x)\cdot (\nabla \psi )(t, \mathbf F_{(t)}(\mathbf x))\} \sigma _0(\mathbf x) \, dtd\mathbf x \\
 &= \int_{[0, \tau ) \times \Omega } \! \partial _t\psi (t, \mathbf X)\sigma (t, \mathbf X) \, dtd\mathbf X.
 \end{align*}
Now, using (\ref{2.39}) and the fact that $\psi (\tau , \cdot ) \equiv 0$ (by compact support), we get
\[ \int_{[0, \tau ) \times \Omega } \! \partial _t[\psi (t, \mathbf F_{(t)}(\mathbf x))]\sigma _0(\mathbf x) \, dtd\mathbf x = -\int_\Omega  \! \psi (0, \mathbf x)\sigma _0(\mathbf x) \, d\mathbf x. \]
Hence,
\[ -\int_\Omega  \! \psi (0, \mathbf x)\sigma _0(\mathbf x) \, d\mathbf x - \int_{[0, \tau ) \times \Omega } \! \partial _t\mathbf F_{(t)}(\mathbf x)\cdot (\nabla \psi )(t, \mathbf F_{(t)}(\mathbf x)) \sigma _0(\mathbf x) \, dtd\mathbf x \]
\[ = \int_{[0, \tau ) \times \Omega } \! \partial _t\psi (t, \mathbf X)\sigma (t, \mathbf X) \, dtd\mathbf X. \]
If we make the change of variables $\mathbf X = \mathbf F_{(t)}(\mathbf x)$ in the second integral above, then, by (\emph{ii}), (\emph{iii}) in Definition \ref{2.4} and by (\ref{appendix}), we have that
\[ -\int_\Omega  \! \psi (0, \mathbf x)\sigma _0(\mathbf x) \, d\mathbf x - \int_{[0, \tau ) \times \Omega } \! (\partial _t\mathbf F_{(t)})(\mathbf F_{(t)}^*(\mathbf X))\cdot \nabla \psi (t, \mathbf X) \sigma (t, \mathbf X) \, dtd\mathbf X \]
\[ = \int_{[0, \tau ) \times \Omega } \! \partial _t\psi (t, \mathbf X)\sigma (t, \mathbf X) \, dtd\mathbf X. \]
Then, rearranging and using the definition of $\mathbf u$ in (\ref{u}), we obtain
\[ \int_{[0, \tau ) \times \Omega } \! \{ \partial _t\psi (t, \mathbf X) + \mathbf u(t, \mathbf X)\cdot \nabla \psi (t, \mathbf X)\} \sigma (t, \mathbf X) \, dtd\mathbf X + \int_\Omega  \! \psi (0, \mathbf x)\sigma _0(\mathbf x) \, d\mathbf x = 0.\]
Changing notations $\mathbf X$ to $\mathbf x$ gives us (\ref{eweak2}).

We now prove that (\ref{eweak1}) also holds.  By the properties of $\mathbf F$ and $\mathbf T$ in Definition \ref{2.4}, we have that $\mathbf Z(t, \mathbf x)$ as defined in (\ref{Z}) satisfies $\mathbf Z \in L^\infty ([0, \tau ] \times \overline{\Omega })$.  Also, applying the definition of $\mathbf Z(t, \mathbf x)$ in (\ref{Z}) to equation (\ref{Zweak}) gives
\begin{equation} \label{Zweak2}
\begin{split}
\int_{[0,\tau )\times \Omega } \! \{\mathbf T(t, \mathbf F_{(t)}(\mathbf x)) \cdot \partial _t &\varphi (t, \mathbf x) + \mathbf e_3 \times [\mathbf T(t, \mathbf F_{(t)}(\mathbf x)) - \mathbf F_{(t)}(\mathbf x)] \cdot \varphi (t, \mathbf x)\} \sigma _0(\mathbf x) \, dtd\mathbf x \\
&+ \int_\Omega \! [\mathbf T_0(\mathbf x)\cdot \varphi (0, \mathbf x)]\sigma _0(\mathbf x) \, d\mathbf x = 0,
\end{split}
\end{equation}
for any $\varphi \in C_c^1([0,\tau ) \times \Omega )$.
Now, since $\Omega $ is a bounded set and $\mathbf F_{(t)} \textrm{\#} \sigma _0(\cdot )  = \sigma (t, \cdot )$ for all $t \in [0, \tau )$, equation (\ref{appendix}) allows us to make the change of variables $\mathbf X = \mathbf F_{(t)}(\mathbf x)$ in the first integral of (\ref{Zweak2}).  Thus, by (\emph{iii}) of Definition \ref{2.4}, we have that $\mathbf x = \mathbf F_{(t)}^*(\mathbf X)$ for $\sigma -a.e.$ $\mathbf x \in \Omega$ for every $t \in [0, \tau )$ and then, from (\ref{Zweak2}), we obtain
\begin{equation} \label{2.41}
\begin{split}
\int_{[0,\tau )\times \Omega } \! \{ \mathbf T(t, \mathbf X) \cdot \partial _t &\varphi (t, \mathbf F_{(t)}^*(\mathbf X)) + \mathbf e_3 \times [\mathbf T(t, \mathbf X) - \mathbf X] \cdot \varphi (t, \mathbf F_{(t)}^*(\mathbf X))\} \sigma (t, \mathbf X) \, dtd\mathbf X \\
&+ \int_\Omega \! [\mathbf T_0(\mathbf x)\cdot \varphi (0, \mathbf x)]\sigma _0(\mathbf x) \, d\mathbf x = 0,
\end{split}
\end{equation}
for any $\varphi \in C_c^1([0,\tau ) \times \Omega )$.
We now show that (\ref{2.41}) also holds for all $\varphi $ such that
\begin{align}
\nonumber &\varphi \in L^\infty ([0,\tau ) \times \Omega ),\\
\label{phicon}\partial &_t\varphi \in L^\infty ([0,\tau ) \times \Omega ),\\
\nonumber supp(\varphi ) \subset &[0, \tau - \epsilon ] \times \overline{\Omega }\quad \textrm{for some} \quad \epsilon > 0
\end{align}
In order to do this, we construct an approximating sequence for such $\varphi $.  Let us extend $\varphi $ to $(-\infty , \infty ) \times \Omega $ by defining, for $\mathbf x \in \Omega $, $\varphi (t, \mathbf x) = \varphi (-t, \mathbf x)$ for $t < 0$ and $\varphi (t, \cdot ) \equiv 0$ for $t \geqslant \tau $.  Then, we let $h > 0$ and define $\Omega _h = \{ \mathbf x \in \Omega : dist(\mathbf x, \partial \Omega ) > h\} $, where $\partial \Omega $ denotes the boundary of $\Omega $.  Thus, $\varphi \chi _{\Omega _h}$ is now defined on $\mathbb{R} ^1 \times \mathbb{R} ^3$, where $\chi _{\Omega _h}$ denotes the characteristic function of the set $\Omega _h$.  Next, let $j _h(t, \mathbf x) = \frac{1}{h^4}j (\frac{|(t, \mathbf x)|}{h})$, where $j (\cdot )$ is a standard mollifier, and let $k > \frac{1}{\epsilon }$ be an integer.  We then have that functions $\varphi _k = (\varphi \chi _{\Omega _{4h}}) * j _h$, with $h = \frac{1}{k} < \epsilon $, satisfy 
\[ \varphi _k \in C_c^1([0, \tau ) \times \Omega )\quad \textrm{with} \quad \left\| \varphi _k, \partial _t\varphi _k\right\|_{L^{\infty ([0, \tau ) \times \Omega )}} \leqslant C,\]
where $C$ does not depend on $k$, and
\[(\varphi _k, \partial _t\varphi _k) \rightarrow (\varphi , \partial _t\varphi ) \quad a.e. \textrm{ on }[0, \tau ) \times \Omega \textrm{ as }k \rightarrow \infty .\]
Thus, by the dominated convergence theorem,
\[ \lim _{k \rightarrow \infty }\Intxt{(\varphi _k, \partial _t\varphi _k)(t, \mathbf x)\sigma _0(\mathbf x)} = \Intxt{(\varphi , \partial _t\varphi )(t, \mathbf x)\sigma _0(\mathbf x)}.\]
Also, since $\varphi \in L^\infty ([0,\tau ) \times \Omega )$ by (\ref{phicon}) and $\mathbf F_{(t)}^* \textrm{\#} \sigma (t, \cdot )  = \sigma _0(\cdot )$ , we have from (\ref{appendix}) that
\[ \Intxt{(\varphi _k, \partial _t\varphi _k)(t, \mathbf x)\sigma _0(\mathbf x)} = \IntXt{(\varphi _k, \partial _t\varphi _k)(t, \mathbf F_{(t)}^*(\mathbf X))\sigma (t, \mathbf X)}\]
and
\[ \Intxt{(\varphi , \partial _t\varphi )(t, \mathbf x)\sigma _0(\mathbf x)} = \IntXt{(\varphi , \partial _t\varphi )(t, \mathbf F_{(t)}^*(\mathbf X))\sigma (t, \mathbf X)}.\]
Hence,
\[ \lim _{k \rightarrow \infty }\IntXt{(\varphi _k, \partial _t\varphi _k)(t, \mathbf F_{(t)}^*(\mathbf X))\sigma (t, \mathbf X)} = \IntXt{(\varphi , \partial _t\varphi )(t, \mathbf F_{(t)}^*(\mathbf X))\sigma (t, \mathbf X)}.\]
Then, since $\Omega $ is bounded, $\mathbf T \in L^\infty ([0, \tau ) \times \Omega )$ and (\ref{2.41}) holds for each $\varphi _k$, it follows that (\ref{2.41}) holds for $\varphi $ satisfying (\ref{phicon}).
Now, let
\begin{equation}\label{phidef}
\varphi (t, \mathbf x) = \eta (t, \mathbf F_{(t)}(\mathbf x)),
\end{equation}
where $\eta \in C_c^1([0, \tau ) \times \Omega )$.  Then, since $supp(\eta ) \subset [0, \tau - \epsilon ]\times \overline{\Omega }$ and $\partial _t \mathbf F \in L^\infty ([0, \tau ) \times \Omega )$ by our additional regularity property, we have that $\varphi $ satisfies the conditions in (\ref{phicon}).  Therefore, (\ref{2.41}) holds for $\varphi $ as defined in (\ref{phidef}).  We also see from (\ref{phidef}) and use of the chain rule that
\[ \partial _t\varphi (t, \mathbf x) = (\partial _t\eta )(t, \mathbf F_{(t)}(\mathbf x)) + \partial _t\mathbf F_{(t)}(\mathbf x) \cdot (\nabla \eta )(t, \mathbf F_{(t)}(\mathbf x)).\]
Thus, by property (\emph{iii}) of Definition \ref{2.4}, we have
\[ \varphi (t, \mathbf F_{(t)}^*(\mathbf X)) = \eta (t, \mathbf X)\]
and
\[ \partial _t\varphi (t, \mathbf F_{(t)}^*(\mathbf X)) = \partial _t\eta (t, \mathbf X) + [(\partial _t\mathbf F)(t, \mathbf F_{(t)}^*(\mathbf X))]\cdot \nabla \eta (t, \mathbf X)\]
for $\sigma -a.e.$ $\mathbf X \in \Omega$ for every $t \in [0, \tau )$.  Then, inserting $\varphi (t, \mathbf x) = \eta (t, \mathbf F_{(t)}(\mathbf x))$ into (\ref{2.41}) gives

\[ \IntXt{\{ \mathbf T(t, \mathbf X)\cdot [\partial _t\eta (t, \mathbf X) + [(\partial _t\mathbf F)(t, \mathbf F_{(t)}^*(\mathbf X))]\cdot \nabla \eta (t, \mathbf X)] + \mathbf e_3 \times [\mathbf T(t, \mathbf X) - \mathbf X] \cdot \eta (t, \mathbf X)\} \sigma (t, \mathbf X)}\]
\[+ \Intx{\mathbf T_0(\mathbf x)\cdot \eta (0, \mathbf F_{(0)}(\mathbf x))\sigma _0(\mathbf x)} = 0,\]
and, using the definition of $\mathbf u$ in (\ref{u}) together with property (\emph{i}) of Definition \ref{2.4}, we obtain

\[ \IntXt{\{ \mathbf T(t, \mathbf X)\cdot [\partial _t\eta (t, \mathbf X) + (\mathbf u(t, \mathbf X)\cdot \nabla )\eta (t, \mathbf X)] + \mathbf e_3 \times [\mathbf T(t, \mathbf X) - \mathbf X] \cdot \eta (t, \mathbf X)\} \sigma (t, \mathbf X)}\]
\[ + \Intx{\mathbf T_0(\mathbf x)\cdot \eta (0, \mathbf x)\sigma _0(\mathbf x)} = 0.\]
Finally, changing notations $\mathbf X$ to $\mathbf x$ and $\eta $ to $\varphi $ gives (\ref{eweak1}), proving statement (\emph{i}).  Then, statement (\emph{ii}) follows directly from (\emph{i}).
\end{proof}

We now state the main result, which we will prove in Sections \ref{dual section} and \ref{physical section}: 
\vspace{4mm}
\begin{theorem}\label{2.2}
Let $\Omega \subset \mathbb{R}^3$ be an open bounded convex set.  Assume that 
\[ \mathbf T_0 \textrm{\#} \sigma _0 \in L^r(\Lambda _0)\]
for $r \in (1, \infty )$, where $\Lambda _0 \subset \mathbb{R}^3$ is compact.  Then, for any $\tau > 0$, there exists a weak Lagrangian solution $(\mathbf F, \mathbf T, \sigma )$ of (\ref{momentum})-(\ref{initial}) in $[0, \tau ) \times \Omega $, where $\mathbf T \in L^\infty ([0, \tau ) \times \Omega )$ and (\ref{Fcts}) is satisfied for any $q \in [1, \infty )$.  Moreover, the function $\mathbf Z(t, \mathbf x)$ defined by (\ref{Z}) satisfies $\mathbf Z(\cdot , \mathbf x) \in W^{1, \infty }([0, \tau ))$ for $\sigma _0-a.e.$ $\mathbf x \in \Omega $, and (\ref{Zequation}) is satisfied, in addition to the weak form (\ref{Zweak}), in the following sense:
\begin{equation}\label{2.43}
\begin{split}
&\partial _t \mathbf Z(t, \mathbf x) = \mathbf e_3 \times \left [\mathbf Z(t, \mathbf x) - \mathbf F(t, \mathbf x)\right ] \hspace{3mm} \textrm{ for $\sigma _0-a.e.$ } \mathbf x \in \Omega \textrm{ and every $t \in [0, \tau )$},\\
&\mathbf Z(0, \mathbf x) = \mathbf T_0(\mathbf x) \hspace{3mm} \textrm{ for $\sigma _0-a.e.$ }\mathbf x \in \Omega .
\end{split}
\end{equation}
\end{theorem}
To prove this theorem, we show the existence of a Lagrangian flow map $\Phi $ in dual variables.  We then define $\mathbf F$ in terms of $\Phi $.

Throughout the next two sections, we will use $\sigma $ and $\nu $ to denote the measures $\sigma (t, \cdot )$ and $\nu (t, \cdot )$ considered at some fixed time $t \in [0, \tau )$.  Thus, if we write $\nu -a.e.$ $\mathbf y \in \mathbb{R}^3 $, for example, then we mean $\nu (t, \cdot )-a.e.$ $\mathbf y \in \mathbb{R}^3$ for a fixed $t \in [0, \tau )$.  

\begin{remark}
In a recent development, Ambrosio {\em et al} in   \cite{ambnew} and \cite{ambnew2} have proven  the existence of weak Eulerian solutions of the incompressible semi-geostrophic equations on the 2-dimensional torus and in a convex 3-dimensional domain.  \end{remark}

\section{Existence of the Lagrangian flow map in dual space}\label{dual section}
\setcounter{equation}{0}
In this section we  mimic the the approximation technique of \cite{feldman}, which  makes use of  the results of \cite{ambrosio}, to show existence of the Lagrangian flow map in dual space. The proof  is based on the sequence of approximating equations defined in \cite{maroofi} for the dual space solution of  the compressible case. 

Let $\Omega $, $\Lambda $ be as in Theorem \ref{5.5}, let $\mathbf T_0$, $\sigma _0$ be as in (\ref{initial}) and let $\nu _0 = \mathbf T_0 \textrm{\#}\sigma _0$.  Then, by Theorem \ref{5.5} there exists a solution $(\mathbf T, \mathbf S, \sigma )$ of the system (\ref{comdual1})-(\ref{comdual6}), with initial data $\nu _0$, satisfying all assertions of Theorem \ref{5.5}.
From \cite[Lemma 3.1]{maroofi} and Theorem \ref{3.3} we have that, for a.e. $\mathbf y \in \Lambda $,
\begin{equation}\label{gradg=gradc}
\nabla g_0(\mathbf y) = \nabla c(\mathbf S(\mathbf y), \mathbf y),
\end{equation}
where $g_0$ is a solution of the Kantorovich dual problem (\ref{26}).  Combining this with the definition of $\mathbf w$ in (\ref{comdual2}), we see that the dual space velocity can be written as
\begin{equation}\label{wg}
\mathbf w =  y_3\mathbf e _3 \times \nabla g_0(\mathbf y).
\end{equation}
Thus, the vector field $\mathbf w$ is divergence free.  In addition, since $c(\mathbf x, \mathbf y) \in C^2(\overline{\Omega }\times \overline{\Lambda})$, there exists some constant $\lambda $ such that
\begin{equation}\label{concavityconstant}
D^2_{\mathbf y\mathbf y}c(\mathbf x, \mathbf y) \leqslant \lambda I
\end{equation}
for all $(\mathbf x, \mathbf y) \in \overline{\Omega }\times \overline{\Lambda}$.
  Furthermore, since the potential temperature $y_3$ is assumed to be bounded, we can use (\ref{wg}) to obtain the following properties of $\mathbf w$:
\begin{itemize}
	\item by the semi-concavity of $g_0$ in $\Lambda $, it follows that $\mathbf w(t, \cdot ) \in BV_{loc}(\mathbb{R}^3)$ for $a.e.$ $t \in (0, \tau )$; 
	\item 
	by Theorem \ref{5.5} we have that $\mathbf w \in L^\infty ([0, \tau )\times \Lambda )$,
	\item by (\ref{gradg=gradc}), (\ref{wg}) and (\ref{concavityconstant}) we have that $|D\mathbf w(t, \cdot )|(\Lambda ) \in L^1_{loc}(0, \tau )$.
\end{itemize}
Since $\nu $ has compact support in $[0, \tau ] \times \mathbb{R}^3$, we can modify $\mathbf w$ away from $\Lambda $ so that the modified function $\widetilde{\mathbf w}$ satisfies

\begin{equation}\label{modbounds}                                                       
\widetilde{\mathbf w} \in L^\infty ([0, \tau ) \times \mathbb{R}^3), \qquad \qquad \widetilde{\mathbf w}(t, \cdot ) \in BV_{loc}(\mathbb{R}^3) \textrm{ for $a.e.$ } t\in (0, \tau ), 
\end{equation}
\begin{equation}\label{moddivfree}
|D\widetilde{\mathbf w}(t, \cdot )|(\mathbb{R}^3) \in L^1_{loc}(0, \tau ), \qquad \qquad \nabla \cdot \widetilde{\mathbf w} (t, \cdot ) = 0 \textrm{ in } \mathbb{R}^3 \textrm{ for every } t \in [0, \tau )
\end{equation}
and
\begin{equation}\label{w=modw}
\mathbf w =\widetilde{\mathbf w} \quad \textrm{ in } \Lambda.
\end{equation}
We construct such a modification as follows.  Following \cite[Section 5]{maroofi}, define $\Lambda := B(0, R) \times (\delta , \frac{1}{\delta})$, where $0<\delta <1$ and $B(0, R)$ represents the open ball of radius $R$ centered at the origin in $\mathbb R^2$.  Defining $\delta $, $R$ as in \cite[(55), (56)]{maroofi} ensures that $supp(\nu )$ is contained in $\Lambda $.  Define $\zeta \in C^\infty (\mathbb{R})$ as $\zeta = 1$ on $\{ |s| < R \} $, $\zeta = 0$ on $\{ |s| > R \} $ and $0 \leqslant \zeta \leqslant 1$ on $\mathbb R$.  Define $\xi \in C^\infty (\mathbb R)$ as $\xi = 1$ on $\{ \delta < s < \frac{1}{\delta} \}$, $\xi = 0$ when $s \leqslant \frac{\delta }{2}$ or when $s\geqslant \frac{2}{\delta}$ and $0\leqslant \xi \leqslant 1$ on $\mathbb R$.

Then, define for $\mathbf y \in \mathbb R^3$
\begin{equation}\label{newmod}
\mathbf M(\mathbf y) = (M_1(\mathbf y), M_2(\mathbf y), M_3(\mathbf y))= (\zeta (|y_1|)y_1, \zeta (|y_2|)y_2, \xi (y_3)y_3)
\end{equation}
and define the modified velocity as
\begin{equation}\label{newmodw}
\widetilde{\mathbf w} = \mathbf e _3 \times (\mathbf M(\mathbf y) - \mathbf S(t, \mathbf M(\mathbf y))) = \xi (y_3)y_3\mathbf e _3 \times \nabla g_0(\mathbf M(\mathbf y)).
\end{equation}
Then $\widetilde{\mathbf w}$ satisfies (\ref{modbounds})-(\ref{w=modw}).  These conditions enable us to apply the theory of \cite{ambrosio} to the transport equation (\ref{comdual1}) with $\mathbf w$ replaced by our modified velocity $\widetilde{\mathbf w}$:

\begin{lemma}\label{2.8}
There exists a unique locally bounded Borel measurable map $\Phi : [0, \tau ) \times \mathbb{R}^3 \rightarrow \mathbb{R}^3$ satisfying
\begin{enumerate}[label=(\textit{\roman{*}}), ref=\textit{(\roman{*})}]
	\item $\Phi (\cdot , \mathbf y) \in W^{1, \infty }([0, \tau ))$ for $\nu _0-a.e.$ $\mathbf y \in \mathbb{R}^3$;
\item $\Phi (0, \mathbf y) = \mathbf y$ for $\nu _0-a.e.$ $\mathbf y \in \mathbb{R}^3$;
\item for $\nu _0-a.e.$ $\mathbf y \in \mathbb{R}^3$,
\begin{equation}\label{phiwt}
\partial _t \Phi (t ,\mathbf y) = \widetilde{\mathbf w}(t, \Phi (t ,\mathbf y));
\end{equation}
\item there exists a Borel map $\Phi ^* :  [0, \tau ) \times \mathbb{R}^3 \rightarrow \mathbb{R}^3$ such that, for every $t \in (0, \tau )$, the map $\Phi _{(t)}^* :  \mathbb{R}^3 \rightarrow \mathbb{R}^3$ is Lebesgue-measure preserving, and such that $\Phi _{(t)}^* \circ \Phi _{(t)} (\mathbf y) = \Phi _{(t)} \circ \Phi _{(t)}^*(\mathbf y) = \mathbf y$ for $\nu -a.e.$ $\mathbf y \in \mathbb{R}^3$;
\item $\Phi (t, \cdot ): \mathbb{R}^3 \rightarrow \mathbb{R}^3$ is a Lebesgue-measure preserving map for every $t \in [0, \tau )$.
\end{enumerate}
\end{lemma}

\begin{proof}
The proof is essentially identical to that of \cite[Lemma 2.8]{feldman}.

\end{proof}

We now show that the image of the flow map $\Phi $ is contained in $\Lambda $, and therefore corresponds to the velocity field $\mathbf w$.

\begin{lemma}\label{2.9}
Let $\Lambda $ be as in Theorem \ref{5.5}.  Let $\Phi $ be the map defined in Lemma \ref{2.8} and let $\mathbf w$ be defined as in (\ref{comdual2}).  
 Then
\begin{equation}\label{phiinlambda}
\Phi (t, \mathbf y) \subset \Lambda \quad \textrm{for }\nu _0 -a.e. \quad \mathbf y \in \mathbf T_0(\Omega ) \textrm{ and every } t \in [0, \tau ).
\end{equation}
In particular,
\begin{equation}\label{phiw}
\partial _t \Phi (t, \mathbf y) = \mathbf w(t, \Phi (t, \mathbf y)) \quad \textrm{for }\nu _0 -a.e. \quad \mathbf y \in \mathbf T_0(\Omega ) \textrm{ and every } t \in [0, \tau ).
\end{equation}
\end{lemma}

\begin{proof}
Firstly note that, since $\mathbf T_0 \textrm{\#}\sigma _0= \nu _0$ and the measures $\sigma _0$ and $\nu _0$ are contained within $\Omega $ and $\Lambda $ respectively, we may assume
\begin{equation}\label{mapsto}
\mathbf T_0 : \Omega \rightarrow \Lambda .
\end{equation}
Therefore, for $\sigma _0-a.e.$ $\mathbf x \in \Omega$, we have that $\mathbf T_0(\mathbf x) \in \Lambda $.

We begin with the vertical component of $\Phi $; $\Phi _3$.  From (\ref{w=modw}) and (\ref{phiwt}), we have that $\partial _t \Phi _3(t, \mathbf y) = \widetilde{w}_3(t, \Phi (t, \mathbf y)) = 0$ for all $\mathbf y \in \mathbb{R}^3$ and every $t \in [0, \tau )$.  Therefore, we have $\Phi _3(t, \mathbf y) = \Phi _3(0, \mathbf y) = y_3$ for $\nu _0-$a.e. $\mathbf y \in \mathbb{R}^3$ and every $t \in [0, \tau )$, where we have used Lemma \ref{2.8} (\emph{ii}).  We therefore conclude that $\delta < \Phi _3(t, \mathbf y) < \frac{1}{\delta }$ for $\nu _0-$a.e. $\mathbf y \in \mathbf T_0(\Omega )$ and every $t \in [0, \tau )$, where $\delta $ is as in Remark \ref{support}.

The proof that the horizontal components of $\Phi $ stay inside $\Lambda $ proceeds in a similar way to that of \cite[Lemma 2.12]{feldman}.

\end{proof}

Finally, we wish to prove that when $(\nu , \mathbf T)$ is a weak solution of (\ref{comdual1})-(\ref{comdual6}) then $\nu $ is a weak Lagrangian solution of the transport equation (\ref{comdual1}), i.e. $\nu $ satisfies the property $\nu  = \Phi \textrm{\#}\nu _0$.
To do this, we use the time-approximation scheme in dual space of  \cite[Section 5]{maroofi}, as well as the following result (see \cite{ambbook}.
\begin{proposition}\label{8.1.8}
Let $\nu (t, \cdot )$ be narrowly continuous Borel probability measures solving the continuity equation
\[ \partial _t\nu + \nabla \cdot (\mathbf v\nu ) = 0\]
with respect to a vector field $\mathbf v$ satisfying
\begin{eqnarray}
&&\label{8.1.2} \int_{[0, \tau ) \times \mathbb{R}^3 } \! |\mathbf v|\nu  \, dtd\mathbf y < + \infty ,\\
&&\label{8.1.7} \int_0^{\tau }\! \left [ \sup _B |\mathbf v| + Lip(\mathbf v, B) \right ] \, dt < + \infty ,
\end{eqnarray}
for every compact set $B \subset \mathbb{R}^3$.  Then, for $\nu _0-a.e. \, \mathbf y \in \mathbb{R}^3$ the characteristic system
\begin{equation}\label{ode1}
\Phi (0, \mathbf y) = \mathbf y, \qquad \frac{d}{dt}\Phi (t, \mathbf y) = \mathbf v(t, \Phi (t, \mathbf y))
\end{equation}
admits a globally defined solution $\Phi (t, \mathbf y)$ in $[0, \tau )$ and
\begin{equation}\label{8.1.19}
\nu (t, \cdot ) = \Phi (t, \cdot ) \textrm{\#} \nu _0(\cdot ) \quad \textrm{ for all } t \in [0, \tau ).
\end{equation}
\end{proposition}


The time approximation scheme is based on the  discretisation of  (\ref{comdual1}),  with time step $h$. One then considers piecewise smooth approximate solutions $\nu _h(\cdot , \cdot )$ to (\ref{comdual1})-(\ref{comdual6}), and corresponding velocities ${\mathbf v}_h$ as defined in \cite[Section 5]{maroofi}.  Using these approximating solutions, we can prove the following result.

\begin{proposition}\label{2.11}
Let $\Omega $, $r$, $\mathbf T_0$ be as in Theorem \ref{5.5}, and let $(\nu , \mathbf T)$ be the weak solution of (\ref{comdual1})-(\ref{comdual6}) as constructed in Theorem \ref{5.5}.  Let $\wt $ be defined by (\ref{w=modw}) and let $\Phi $ be the regular Lagrangian flow of $\wt $ defined in Lemma \ref{2.8}.  
Then, for every $t \in [0, \tau ]$,
\begin{equation}\label{nupush}
\nu  = \Phi _{(t)}\textrm{\#}\nu _0.
\end{equation}
Moreover, for every $t \in [0, \tau ]$,
\begin{equation}\label{2.61}
\nu (\mathbf y) = \nu _0(\Phi _{(t)}^*(\mathbf y)) \quad \textrm{for }\nu -a.e.\quad \mathbf y\in \mathbb{R}^3,
\end{equation}
where the map $\Phi _{(t)}^*$ is defined in Lemma \ref{2.8} (iv).
\end{proposition}

\begin{proof}

For $h$, $k$, $j_h$, $g_h^k$ as in \cite[Section 5]{maroofi} define
\begin{equation}\label{wthk}
\wt _h^k(\mathbf y) := \xi (y_3) y_3\mathbf e _3 \times \nabla (j_h * g_h^k)(\mathbf M(\mathbf y)),
\end{equation}
where $\mathbf M$, $\xi $ are defined in (\ref{newmod}).
Define functions $\wt _h$ on $[0, \tau ] \times \mathbb{R}^3$ by setting them equal to $\wt _h^k$ on the time-interval $t \in [kh, (k+1)h)$.  
Following \cite[Lemma 5.3]{maroofi},  the corresponding potential density $\nu _h$ is a weak solution of
\begin{equation}\label{wthcont}
\begin{split}
&\partial _t\nu _h + \nabla \cdot (\nu _h \wt _h)=0 \quad \textrm{ in } (0, \tau ) \times \mathbb{R}^3,\\
&\nu _h(0, \mathbf y) = \nu _h^0(\mathbf y).
\end{split}
\end{equation}

The construction of  $\wt _h$ implies that  $\wt _h$ is a divergence-free vector field satisfying (\ref{modbounds})-(\ref{w=modw}) and
\begin{equation}
\begin{split}\label{ambprop}
& (\emph{i}) \qquad \int_{[0, \tau ) \times \mathbb{R}^3 } \! |\wt _h|\nu _h \, dtd\mathbf y < + \infty ,\\
& (\emph{ii}) \qquad \int_0^{\tau }\! \left [ \sup _B |\wt _h| + Lip(\wt _h, B) \right ] \, dt < + \infty ,
\end{split}
\end{equation}
 Thus, by \cite[Section 6]{ambrosio}, there exists a unique Lagrangian flow $\Phi _h : \mathbb{R}^3 \times \mathbb{R} \rightarrow \mathbb{R}^3$ induced by $\wt _h$ and, for each $t$, the map $(\Phi _h)_{(t)} : \mathbb{R}^3 \rightarrow \mathbb{R}^3$ is $\mathcal{L}^3$-measure preserving.

As in  \cite[Section 5]{maroofi}, we can find a decreasing sequence $\{ h_j\} $ converging to $0$ such that
\begin{eqnarray}
&&\nu _{h_j} \rightarrow \nu \qquad \textrm{ weakly in } L^r((0, \tau ) \times \mathbb{R}^3),\nonumber \\
&&\label{converge}\nu _{h_j} \wt _{h_j} \rightarrow \nu \mathbf w \qquad \textrm{ weakly in } L^r((0, \tau ) \times \mathbb{R}^3),\\
&&\nu _{h_j}(t, \cdot ) \rightarrow \nu (t, \cdot ) \qquad \textrm{ weakly in } L^r(\mathbb{R}^3),\textrm{ for all } t \in [0, \tau ].\nonumber
\end{eqnarray}
Thus, since we have compact support, we can use (\ref{w=modw})  and the dominated convergence theorem to obtain
\begin{equation}\label{wthconvergence}
\wt _{h_j} \rightarrow \wt \qquad \textrm{ weakly in } L^r((0, \tau ) \times \Lambda ).
\end{equation}
Finally, we can use \cite[Theorem 6.6]{ambrosio} to conclude that  for each $t \in [0, \tau ]$,
\begin{equation}\label{phiconvergence}
(\Phi _{h_j} )_{(t)} \rightarrow \Phi _{(t)} \qquad \textrm{ in } L^1_{loc}(\mathbb{R}^3)
\end{equation}
as $j \rightarrow \infty $.

Since $\nu _h$ is narrowly continuous, we may combine (\ref{wthcont}) with (\ref{ambprop}) and thus use Proposition \ref{8.1.8} to conclude that the system
\begin{equation*}
\frac{d}{dt} \Phi _h(t, \mathbf y) = \wt _h(t, \Phi _h(t, \mathbf y)),
\end{equation*}
with initial condition $\Phi (0, \mathbf y) = \mathbf y$, admits a globally defined solution $\Phi _h(t, \mathbf y)$ and
\begin{equation}\label{nuhpush}
\nu _h(t, \cdot ) = (\Phi _h)_{(t)}\textrm{\#}\nu _h^0.
\end{equation}

Using (\ref{appendix}) and the properties of $\nu_h$, we obtain for any $t>0$, $j=1, ...$ and any $\varphi \in C_c(\mathbb{R}^3)$
\[ \int_{\mathbb{R}^3} \! \varphi (\Phi _{h_j}(t, \mathbf y))\nu ^0_{h_j}(\mathbf y) \, d \mathbf y = \int_{\mathbb{R}^3} \! \varphi (\mathbf Y)\nu _{h_j}(t, \mathbf Y) \, d \mathbf Y.\]
Passing to the limit $j \rightarrow \infty $ in the last equality, using (\ref{phiconvergence}), the fact that $\nu ^0_{h_j} \rightarrow \nu _0$ as $j\rightarrow \infty $ in $L^r(\mathbb{R}^3)$ and the dominated convergence theorem in the left-hand side, and using (\ref{converge}) in the right-hand side, we obtain
\begin{equation}\label{2.71}
\int_{\mathbb{R}^3} \! \varphi (\Phi (t, \mathbf y))\nu _0(\mathbf y) \, d \mathbf y = \int_{\mathbb{R}^3} \! \varphi (\mathbf Y)\nu (t, \mathbf Y) \, d \mathbf Y
\end{equation}
for any $\varphi \in C_c(\mathbb{R}^3)$.  This implies (\ref{nupush}).

Since $\Phi _{(t)}$ is a measure preserving map, we use Lemma \ref{2.8} (\emph{iv}) to conclude that the left-hand side of (\ref{2.71}) is equal to
\[ \int_{\mathbb{R}^3} \! \varphi (\mathbf Y)\nu _0(\Phi _{(t)}^*(\mathbf Y)) \, d \mathbf Y,\]
and now (\ref{2.71}) implies (\ref{2.61}).

\end{proof}


\begin{remark}
It would be desirable to be able to avoid using the approximating  solutions  in dual space when showing that  $\nu $ is a weak Lagrangian solution of the transport equation.  However, we have not been able to  approximate in $L^1$ the velocity  $\wt $ directly. We appeal instead to the sequence of  solutions of the approximating equations in dual space constructed in \cite{maroofi}, as was done for the proof of the analogous result for the incompressible case given in \cite{feldman}.
\end{remark}

\section{Lagrangian flow in physical space}\label{physical section}
\setcounter{equation}{0}
Throughout this section we will assume that $\Omega $, $\Lambda $, $r$, $\mathbf T_0$, $\mathbf T$, $\nu $, $\mathbf w$, $\wt $, $\Phi $ are as in Proposition \ref{2.11}.  Note that, by Theorem \ref{5.5}, we can apply (\ref{appendix}) to $\sigma _0$, $\sigma $, $\nu _0$, $\nu $ throughout this section.  

We now perform the last step of the analysis and prove the existence of a Lagrangian flow $\mathbf F : [0, \tau ) \times \Omega \rightarrow \Omega $ in the physical space.  Indeed, we define $\mathbf F_{(t)} : \Omega \rightarrow \Omega $ for $t \in [0, \tau )$ as
\begin{equation}\label{phyF}
\mathbf F_{(t)} := \mathbf S_{(t)} \circ \Phi _{(t)} \circ \mathbf T_0,
\end{equation}
where $\mathbf T_0$ is as in (\ref{initial}), $\mathbf S_{(t)}$ is the inverse of $\mathbf T_{(t)}$ (see Theorem \ref{3.3}) and $\Phi _{(t)}$ is the Lagrangian flow in dual space constructed in Lemma \ref{2.8}.  To justify this definition, we prove the following lemma:

\begin{lemma}\label{2.12}
For any $t \in [0, \tau )$, the right hand side of (\ref{phyF}) is defined $\sigma _0-a.e.$ in $\Omega $.  The map $\mathbf F :[0, \tau ) \times \Omega \rightarrow \Omega $ defined by (\ref{phyF}) is Borel.
\end{lemma}

\begin{proof}
Since $\mathbf T_0$ exists and is unique $\sigma _0-a.e.$ in $\Omega $, we have that $\mathbf T_0$ exists and is unique on $\Omega \setminus N_0^1$ where $N_0^1$ is a Borel subset of $\Omega $ with $\sigma _0[N_0^1] = 0$.  Also, since $\mathbf S$ exists and is unique $\nu -a.e.$ in $\Lambda $ for every $t \in [0, \tau )$, we have that $\mathbf S$ exists and is unique on $\Lambda \setminus N^2$ for every $t \in [0, \tau )$, where $N^2$ is a Borel subset of $\Lambda $ with $\nu [N^2] = 0$.  
Then, the right-hand side of (\ref{phyF}) is defined for all
\[ \mathbf x \in \Omega  \setminus (N_0^1 \cup M),\]
where
\[ M = \left\{ \mathbf X \in (\Omega \setminus N_0^1) \textrm{ } : \textrm{ } \Phi _{(t)}(\mathbf T_0(\mathbf X)) \in N^2 \right\} .\]
Note that, from its definition, $M$ is a Borel set.

It remains to prove that $\sigma _0[M] = 0$ for every $t \in [0, \tau )$ 

Fix $t \in [0, \tau )$.  
 Then, using that $\mathbf T_0\textrm{\#}\sigma _0 = \nu _0$ and thus $\mathbf T_0\textrm{\#}\sigma _0 = \nu _0$ for all $\mathbf x \in \Omega \setminus N_0^1$, and using (\ref{nupush}) as well as Lemma \ref{2.8} (\emph{iv}),  we can apply (\ref{appendix}) and compute
\begin{eqnarray*}
\sigma _0  \left [ M_{(t)} \right ] &&=  \sigma _0  \left [ \left\{ \mathbf x \in \Omega \setminus N_0^1 \, : \,  \mathbf T_0(\mathbf x) \in \Phi _{(t)}^{-1}(N^2)\right\} \right ]\\
&&= \int_{\mathbf T_0^{-1}(\Phi _{(t)}^{-1}(N^2)) } \! \sigma _0(\mathbf x) \, d \mathbf x = \int_{\Phi _{(t)}^{-1}(N^2) } \! \nu _0(\mathbf y) \, d \mathbf y\\
&&= \int_{N^2 } \! \nu (\mathbf y) \, d \mathbf y = 0.
\end{eqnarray*}
 Thus, we can define $\mathbf F : [0, \tau ) \times \Omega \rightarrow \Omega $ by (\ref{phyF}).  Then, by Lemma \ref{2.8}, $\mathbf F$ is a Borel mapping.
\end{proof}

It remains to prove that, if $\mathbf F$ is defined by (\ref{phyF}), then $(\mathbf F, \mathbf T, \sigma )$ is a weak Lagrangian solution of (\ref{momentum})-(\ref{initial}) in the sense of Definition \ref{2.4}.  We begin by showing that the initial condition for the flow in Definition \ref{2.4} (\emph{i}) is satisfied.

\begin{proposition}\label{2.13}
Let $\mathbf F$ be defined as in (\ref{phyF}).  Then, $\mathbf F(0, \mathbf x) = \mathbf x$ for $\sigma _0-a.e.$ $\mathbf x \in \Omega $.
\end{proposition}

\begin{proof}
By (\ref{phyF}) we have that $\mathbf F_{(0)}(\mathbf x) = \mathbf S_{(0)} \circ \Phi _{(0)} \circ \mathbf T_0(\mathbf x)$ for all $\mathbf x \in \Omega \setminus N_0$ where $N_0$ is a Borel set with $\sigma _0[N_0]=0$.

By Lemma \ref{3.3}, there exist Borel sets $N_1 \subset \Omega $, $N_2 \subset \Lambda $ with $\sigma _0[N_1] = \nu _0[N_2] = 0$ such that $\mathbf T_0$ exists and is unique in $\Omega \setminus N_1$ and $\mathbf S_{(0)}$ exists and is unique in $\Lambda \setminus N_2$.  Moreover, if $\mathbf x \in \Omega \setminus [N_1 \cup \mathbf T_0^{-1}(N_2)]$, then $\mathbf S_{(0)} \circ \mathbf T_0 (\mathbf x) = \mathbf x$.  Also, by Lemma \ref{2.8} (\emph{ii}), we have that $\Phi _{(0)}(\mathbf y) = \mathbf y$ in $\Lambda \setminus N_3$, where $N_3$ is a Borel set with $\nu _0[N_3] = 0$.

Therefore, $\mathbf F_{(0)}(\mathbf x) = \mathbf x$ for all $\mathbf x \in \Omega \setminus [N_0 \cup N_1 \cup \mathbf T_0^{-1}(N_2 \cup N_3)]$.  We must now show that $\sigma _0\left [ \Omega \cap \mathbf T_0^{-1}(N_2 \cup N_3)\right ] = 0$.

We have that $\nu _0 [N_2 \cup N_3] = 0$.  Then, since $\mathbf T_0 \textrm{\#} \sigma _0 = \nu _0$, we obtain
\[ \sigma _0\left [ \Omega \cap \mathbf T_0^{-1}(N_2 \cup N_3)\right ] = \int_{\mathbf T_0^{-1}(N_2 \cup N_3)} \! \sigma _0(\mathbf x) \, d \mathbf x = \int_{N_2 \cup N_3} \! \nu _0(\mathbf y) \, d \mathbf y= \nu _0 [N_2 \cup N_3] = 0.\]
\end{proof}
 
Next, we prove that property (\emph{ii}) of Definition \ref{2.4} is satisfied.

\begin{proposition}\label{2.14}
For every $t > 0$, the map $\mathbf F_{(t)} : \Omega \rightarrow \Omega $ as defined in (\ref{phyF}) satisfies $\mathbf F_{(t)} \textrm{\#} \sigma _0 = \sigma $.
\end{proposition}

\begin{proof}
In order to prove that $\mathbf F_{(t)} \textrm{\#} \sigma _0 = \sigma $, we must show that, for any $\varphi \in C(\mathbb{R}^3)$,
\begin{equation*}
\int_{\Omega } \! \varphi (\mathbf F_{(t)}(\mathbf x))\sigma _0(\mathbf x) \, d \mathbf x = \int_{\Omega } \! \varphi (\mathbf X)\sigma (\mathbf X) \, d \mathbf X.
\end{equation*}
Then the result will follow from (\ref{appendix}). 

Let $\varphi \in C(\mathbb{R}^3)$.  From the definition of $\mathbf F_{(t)}$ in (\ref{phyF}) we have
\[ \int_{\Omega } \! \varphi (\mathbf F_{(t)}(\mathbf x))\sigma _0(\mathbf x) \, d \mathbf x = \int_{\Omega } \! \varphi \circ \mathbf S_{(t)} \circ \Phi _{(t)} \circ \mathbf T_0(\mathbf x) \sigma _0(\mathbf x) \, d \mathbf x.\]
Then, using $\mathbf T_0 \textrm{\#} \sigma _0 = \nu _0$ 
 we apply (\ref{appendix}) to obtain
\[ \int_{\Omega } \! \varphi \circ \mathbf S_{(t)} \circ \Phi _{(t)} \circ \mathbf T_0(\mathbf x) \sigma _0(\mathbf x) \, d \mathbf x = \int_{\Lambda } \! \varphi \circ \mathbf S_{(t)} \circ \Phi _{(t)} (\mathbf y) \nu _0(\mathbf y) \, d \mathbf y,\]
since $\varphi \circ \mathbf S_{(t)} \circ \Phi _{(t)} \in L^\infty (\Lambda )$.  Then, since $\varphi \circ \mathbf S_{(t)} \in L^\infty (\Lambda)$ 
 we can use (\ref{nupush}) and apply (\ref{appendix}) to get
\[ \int_{\Lambda } \! \varphi \circ \mathbf S_{(t)} \circ \Phi _{(t)} (\mathbf y) \nu _0(\mathbf y) \, d \mathbf y = \int_{\Lambda } \! \varphi \circ \mathbf S_{(t)}(\mathbf Y)\nu (\mathbf Y) \, d \mathbf Y.\]
Finally, since $\mathbf S_{(t)}$ satisfies $\mathbf S_{(t)}\textrm{\#} \nu  = \sigma $, we have that
\[\int_{\Lambda  } \! \varphi \circ \mathbf S_{(t)}(\mathbf Y)\nu (\mathbf Y) \, d \mathbf Y = \int_{\Omega } \! \varphi (\mathbf X)\sigma (\mathbf X) \, d \mathbf X.\]

Thus, we have shown that
\begin{eqnarray*}
\int_{\Omega } \! \varphi (\mathbf F_{(t)}(\mathbf x))\sigma _0(\mathbf x) \, d \mathbf x &&= \int_{\Omega } \! \varphi \circ \mathbf S_{(t)} \circ \Phi _{(t)} \circ \mathbf T_0(\mathbf x) \sigma _0(\mathbf x) \, d \mathbf x\\
&&= \int_{\Lambda  } \! \varphi \circ \mathbf S_{(t)} \circ \Phi _{(t)} (\mathbf y) \nu _0(\mathbf y) \, d \mathbf y\\
&&= \int_{\Lambda  } \! \varphi \circ \mathbf S_{(t)}(\mathbf Y)\nu (\mathbf Y) \, d \mathbf Y\\
&&= \int_{\Omega } \! \varphi (\mathbf X)\sigma (\mathbf X) \, d \mathbf X,
\end{eqnarray*}
as required.
\end{proof}

We now prove that (\ref{Fcts}) holds for all $q \in [1, \infty )$.

\begin{proposition}\label{2.15}
For any $t_0 \in [0, \tau )$ and any $q \in [1, \infty )$,
\[ \lim _{t \rightarrow t_0, t \in [0, \tau )} \int_{\Omega } \! |\mathbf F_{(t)}(\mathbf x) - \mathbf F_{(t_0)}(\mathbf x)|^q \sigma _0(\mathbf x)\, d \mathbf x = 0.\]
\end{proposition}

\begin{proof}
By Lemma \ref{2.12} we have that, for any $t \in [0, \tau )$, (\ref{phyF}) holds $\sigma _0-a.e.$ in $\Omega$.  Thus, since $\mathbf T_0 \textrm{\#} \sigma _0 = \nu _0$, we see that, for any $t, t_0 \in [0, \tau )$,
\begin{eqnarray*}
\int_{\Omega } \! &&\left| \mathbf F_{(t)}(\mathbf x)- \mathbf F_{(t_0)}(\mathbf x) \right| ^q \sigma _0(\mathbf x) \, d \mathbf x = \int_{\Omega } \! \left| \mathbf S_{(t)} \circ \Phi _{(t)} \circ \mathbf T_0(\mathbf x)-\mathbf S_{(t_0)} \circ \Phi _{(t_0)} \circ \mathbf T_0(\mathbf x)\right| ^q \sigma _0(\mathbf x) \, d \mathbf x\\
&&= \int_{\Lambda  } \! \left| \mathbf S_{(t)} \circ \Phi _{(t)}(\mathbf y) - \mathbf S_{(t_0)} \circ \Phi _{(t_0)}(\mathbf y)\right| ^q \nu _0(\mathbf y) \, d \mathbf y\\
&&= \int_{\Lambda  } \! \left| \mathbf S_{(t)} \circ \Phi _{(t)}(\mathbf y) - \mathbf S_{(t_0)} \circ \Phi _{(t)}(\mathbf y) + \mathbf S_{(t_0)} \circ \Phi _{(t)}(\mathbf y) - \mathbf S_{(t_0)} \circ \Phi _{(t_0)}(\mathbf y)\right| ^q \nu _0(\mathbf y) \, d \mathbf y\\
&&\leqslant C\int_{\Lambda  } \! \left| \mathbf S_{(t)} \circ \Phi _{(t)}(\mathbf y) - \mathbf S_{(t_0)} \circ \Phi _{(t)}(\mathbf y)\right| ^q \nu _0(\mathbf y) \, d \mathbf y \\
&&\quad + C\int_{\Lambda  } \! \left| \mathbf S_{(t_0)} \circ \Phi _{(t)}(\mathbf y) - \mathbf S_{(t_0)} \circ \Phi _{(t_0)}(\mathbf y)\right| ^q \nu _0(\mathbf y) \, d \mathbf y\\
&&=: C(I_1 + I_2).
\end{eqnarray*}

Firstly, we show that $I_1 \rightarrow 0$ as $t \rightarrow t_0$.  Note that, using (\ref{2.61}) and Lemma \ref{2.8} (\emph{iv}) 
 we have that $\left\| \nu _{(t)}\right\| _{L^q(\Lambda )} = \left\| \nu _0 \right\| _{L^q(\Lambda )}$ for $t \in (0, \tau )$, so that  $\left\| \nu _{(t)}\right\| _{L^q(\Lambda )}$ is independent of $t$.  Let $r$ and $r'$ be conjugate exponents (i.e. $\frac{1}{r} + \frac{1}{r'} = 1$), with $1 < r < \infty $.  Then, we can use (\ref{nupush}) and H\"{o}lder's inequality to estimate
\begin{eqnarray*}
I_1 &&= \int_{\Lambda  } \! \left| \mathbf S_{(t)} \circ \Phi _{(t)}(\mathbf y) - \mathbf S_{(t_0)} \circ \Phi _{(t)}(\mathbf y)\right| ^q \nu _0(\mathbf y) \, d \mathbf y \\
&&= \int_{\Lambda  } \! \left| \mathbf S_{(t)}(\mathbf y) - \mathbf S_{(t_0)}(\mathbf y)\right| ^q \nu _{(t)} (\mathbf y) \, d \mathbf y \\
&&\leqslant \left \{ \int_{\Lambda  } \! \left| \mathbf S_{(t)}(\mathbf y) - \mathbf S_{(t_0)}(\mathbf y)\right| ^{qr'} \, d \mathbf y \right \} ^{\frac{1}{r'}} \left \{ \int_{\Lambda  } \! \left| \nu _{(t)} (\mathbf y) \right| ^r  \, d \mathbf y \right \} ^{\frac{1}{r}}\\
&&= \left\| \mathbf S_{(t)} - \mathbf S_{(t_0)}\right\|^q_{qr'}\left\| \nu _{(t)}\right\|_r\\
&&= \left\| \mathbf S_{(t)} - \mathbf S_{(t_0)}\right\|^q_{qr'}\left\| \nu _0 \right\|_r \rightarrow 0 \textrm{ as } t \rightarrow t_0.
\end{eqnarray*}


Next, we show that $I_2 \rightarrow 0$ as $t \rightarrow t_0$.  Since $\mathbf S_{(t)} \in \Omega $ for each $t$ and for $\nu -a.e.$ $\mathbf y$, 
 then, by the dominated convergence theorem, it remains to prove that for every $t_0$,
\begin{equation}\label{2.76}
\mathbf S_{(t_0)} \circ \Phi _{(t)}(\mathbf y) - \mathbf S_{(t_0)} \circ \Phi _{(t_0)}(\mathbf y) \rightarrow 0 \quad \textrm{ as } t \rightarrow t_0
\end{equation}
for $\nu -a.e.$ $\mathbf y \in \Lambda $.  First we note that, since $\Phi _{(t)}$ is measure preserving, then it follows from Lemma \ref{2.8} (\emph{i}), and the fact that $\wt \in L^\infty ([0, \tau )\times \mathbb{R}^3)$ by (\ref{modbounds}), that
\[ \Phi _{(t)}(\mathbf y) \rightarrow \Phi _{(t_0)}(\mathbf y) \quad \textrm{ as } t \rightarrow t_0\]
in $[0, \tau ]$ for $\nu -a.e.$ $\mathbf y \in \Lambda $.  If $\mathbf y$ is such a point and if, in addition, $\Phi _{(t_0)}(\mathbf y)$ is a point of continuity for $\mathbf S_{(t_0)}$, then convergence in (\ref{2.76}) holds at $\mathbf y$.  Since $\Phi _{(t_0)}$ is measure preserving, it follows that $\Phi _{(t_0)}(\mathbf y)$ is a point of continuity for $\mathbf S_{(t_0)}$ for $\nu -a.e.$ $\mathbf y$.  Thus, (\ref{2.76}) holds for $\nu -a.e.$ $\mathbf y \in \Lambda $.
\end{proof}

\begin{lemma}\label{2.16}
Let $\mathbf Z$ be defined as in (\ref{Z}) with $\mathbf F$ defined as in (\ref{phyF}).  Then, for all $t \in [0, \tau )$,
\[ \mathbf Z_{(t)} (\mathbf x) = \Phi _{(t)} \circ \mathbf T_0(\mathbf x)\]
for $\sigma _0-a.e.$ $\mathbf x \in \Omega $.
\end{lemma}

\begin{proof}
Using (\ref{Z}), we have that $\mathbf Z_{(t)} = \mathbf T_{(t)} \circ \mathbf F_{(t)}$.  Therefore, we need to justify the following formal computation:
\[ \mathbf T_{(t)} \circ \mathbf F_{(t)} = \mathbf T_{(t)} \circ \mathbf S_{(t)} \circ \Phi _{(t)} \circ \mathbf T_0 = \Phi _{(t)} \circ \mathbf T_0 \]
since, by (\ref{comdual3}), $\mathbf T_{(t)} \circ \mathbf S_{(t)}$ is the identity on the support of $\nu $.  



Now we make this argument rigorous.  Since $\mathbf T$ exists and is unique $\sigma -a.e.$ in $\Omega $ for every $t \in [0, \tau )$, we have that $\mathbf T$ exists and is unique in $\Omega \setminus N^1$ for every $t \in [0, \tau )$, where $N^1$ is a Borel subset of $\Omega $ with $\sigma [N^1] = 0$. 
Then, by Proposition \ref{2.14} we have that
\[ \sigma _0 \left [ \Omega \cap \mathbf F_{(t)}^{-1}(N^1) \right ] = \int_{\mathbf F_{(t)}^{-1}(N^1)} \! \sigma _0(\mathbf x) \, d \mathbf x = \int_{N^1} \! \sigma (\mathbf X) \, d \mathbf X = 0.\]
Now, using Lemma \ref{2.12}, we conclude that
\[ \mathbf Z_{(t)} (\mathbf x) = \mathbf T_{(t)} \circ \mathbf S_{(t)} \circ \Phi _{(t)} \circ \mathbf T_0(\mathbf x)\]
for $\mathbf x \in \Omega \setminus \widetilde{N}$ where $\sigma _0[\widetilde{N}] = 0$.  

Let $\widetilde{M} = \{ \mathbf y  \in \Lambda  \textrm{ } : \textrm{ } \mathbf T(t, \mathbf S(t, \mathbf y)) \neq \mathbf y\}$, 
 Then $\widetilde{M}$ is a Borel set.  Now, the proof of the lemma will be completed if we show that
\begin{equation}\label{2.77}
\sigma _0\left [ \{ \mathbf x \in \Omega \setminus \widetilde{N} \textrm{ } : \textrm{ } \Phi _{(t)} \circ \mathbf T_0(\mathbf x) \in \widetilde{M}\}\right ] = 0.
\end{equation}

From Theorem \ref{3.3} (\emph{iii}) we have that, for any $t \in [0, \tau )$,
\[ \mathbf T_{(t)} \circ \mathbf S_{(t)} (\mathbf y) = \mathbf y \quad \textrm{ for } \nu -a.e. \textrm{ } \mathbf y \in \Lambda .\]
Thus, we have 
\[\int_{\widetilde{M}} \! \nu (\mathbf y) \, d \mathbf y = 0\]
for any $t \in [0, \tau )$.  Therefore, using that $\sigma _0[\widetilde{N}] = 0$, which implies that $\mathbf T_0 \textrm{\#} \sigma _0 = \nu _0$ for all $\mathbf x \in \Omega \setminus \widetilde{N}$, and also using 
(\ref{nupush}), we obtain for any $t \in [0, \tau )$
\begin{eqnarray*}
\sigma _0 \left [ \{ \mathbf x \in \Omega \setminus \widetilde{N} \textrm{ } : \textrm{ } \Phi _{(t)} \circ \mathbf T_0(\mathbf x) \in \widetilde{M} \} \right ] &&= \sigma _0 \left [ \{ \mathbf x \in \Omega \setminus \widetilde{N} \textrm{ } : \textrm{ } \mathbf T_0(\mathbf x) \in \Phi ^*_{(t)}(\widetilde{M}) \} \right ] \\
&&= \int_{\mathbf T_0^{-1}(\Phi ^{-1}_{(t)}(\widetilde{M}))} \! \sigma _0(\mathbf x) \, d \mathbf x\\
&&= \int_{\Phi ^{-1}_{(t)}(\widetilde{M}) } \! \nu _0(\mathbf y) \, d \mathbf y \\
&&= \int_{\widetilde{M}} \! \nu (\mathbf y) \, d \mathbf y = 0,
\end{eqnarray*}
as required.
\end{proof}

We now show existence of the map $\mathbf F^*$ from Definition \ref{2.4} (\emph{iii}).

\begin{proposition}\label{2.17}
The map $\mathbf F$ as defined in (\ref{phyF}) satisfies property (\emph{iii}) of Definition \ref{2.4}.
\end{proposition}

\begin{proof}
As with Lemma \ref{2.12}, we can show that for every $t \in [0, \tau )$ the expression $\mathbf S_{(0)} \circ \Phi ^*_{(t)} \circ \mathbf T_{(t)}(\mathbf x)$ is defined for $\sigma -a.e.$ $\mathbf x \in \Omega $, and the map
\[ \mathbf F^*_{(t)} = \mathbf S_{(0)} \circ \Phi ^*_{(t)} \circ \mathbf T_{(t)} \]
is Borel: 
\begin{quote}Since $\mathbf T$ exists and is unique $\sigma -a.e.$ in $\Omega $ for every $t \in [0, \tau )$, we have that $\mathbf T$ exists and is unique on $\Omega \setminus N^1$ for every $t \in [0, \tau )$, where $N^1$ is a Borel subset of $\Omega $ with $\sigma [N^1] = 0$.  
Also, since $\mathbf S_{(0)}$ exists and is unique $\nu _0-a.e.$ on $\Lambda $, we have that $\mathbf S_{(0)}$ exists and is unique on $\Lambda \setminus N_0^2$, where $N_0^2$ is a Borel subset of $\Lambda $ with $\nu _0[N_0^2] = 0$.  Then, we have that $\mathbf S_{(0)} \circ \Phi ^*_{(t)} \circ \mathbf T_{(t)}(\mathbf x)$ is defined for all
\[ \mathbf x \in \Omega \setminus (N^1 \cup M),\]
where
\[ M = \left\{ \mathbf X \in (\Omega \setminus N^1) \textrm{ } : \textrm{ } \Phi ^*(t, \mathbf T_{(t)}(\mathbf X)) \in N_0^2 \right \} .\]
Note that, from its definition, $M$ is a Borel set.  We must now show that $\sigma [M] = 0$ for every $t \in [0, \tau )$. 
Fix $t \in [0, \tau )$.  
 Then, using that $\mathbf T_{(t)}\textrm{\#}\sigma = \nu $ and thus $\mathbf T_{(t)}\textrm{\#}\sigma = \nu $ for all $\mathbf x \in \Omega \setminus N^1$, and using (\ref{nupush}), we can apply (\ref{appendix}) and compute
\begin{eqnarray*}
\sigma  \left [ M \right ] &&=  \sigma  \left [ \left\{ \mathbf x \in \Omega \setminus N^1 \, : \,  \mathbf T_{(t)}(\mathbf x) \in \Phi _{(t)}(N_0^2)\right\} \right ]\\
&&= \int_{\mathbf T_{(t)}^{-1}(\Phi _{(t)}(N_0^2)) } \! \sigma (\mathbf x) \, d \mathbf x = \int_{\Phi _{(t)}(N_0^2) } \! \nu (\mathbf y) \, d \mathbf y\\
&&= \int_{N_0^2} \! \nu _0(\mathbf y) \, d \mathbf y = 0.
\end{eqnarray*}
Thus, we can define $\mathbf F^* _{(t)}= \mathbf S_{(0)} \circ \Phi ^*_{(t)} \circ \mathbf T_{(t)}$ and $\mathbf F^*$ is a Borel mapping.
\end{quote}

We can now prove that property (\emph{iii}) of Definition \ref{2.4} holds.  Since $\mathbf F_{(t)} \textrm{\#}\sigma _0 = \sigma $, we have that $\mathbf F^*_{(t)} \circ \mathbf F_{(t)}(\mathbf x) = \mathbf S_{(0)} \circ \Phi ^*_{(t)} \circ \mathbf T_{(t)} \circ \mathbf F_{(t)}(\mathbf x)$ for $\sigma -a.e.$ $\mathbf x \in \Omega $.  Then, using Lemma \ref{2.16}, we get $\mathbf F^*_{(t)} \circ \mathbf F_{(t)}(\mathbf x) = \mathbf S_{(0)} \circ \Phi ^*_{(t)} \circ \Phi _{(t)} \circ \mathbf T_0 (\mathbf x)$ $\sigma _0-a.e.$ in $\Omega $.  Since, by Lemma \ref{2.8} (\emph{iv}), $\Phi ^*_{(t)} \circ \Phi _{(t)}(\mathbf y) = \mathbf y$ for $\nu -a.e.$ $\mathbf y$ and thus for $\nu _0 -a.e.$ $\mathbf y \in \Lambda $, and since $\mathbf T_0 \textrm{\#}\sigma _0 = \nu _0$, we have $\Phi ^*_{(t)} \circ \Phi _{(t)} \circ \mathbf T_0 (\mathbf x) = \mathbf T_0(\mathbf x)$ for $\sigma _0-a.e.$ $\mathbf x \in \Omega $.  Thus,  $\mathbf F^*_{(t)} \circ \mathbf F_{(t)}(\mathbf x) = \mathbf S_{(0)} \circ \mathbf T_0(\mathbf x) = \mathbf x$ for $\sigma -a.e.$ $\mathbf x \in \Omega $ by Lemma \ref{3.3} (\emph{iii}).  

By a similar argument, we we have that $\mathbf F_{(t)} \circ \mathbf F^*_{(t)} = \mathbf x$ for $\sigma -a.e.$ $\mathbf x \in \Omega $.
\end{proof}

Finally, we show that property (\emph{iv}) of Definition \ref{2.4} holds for $\mathbf F$ defined in (\ref{phyF}).

\begin{proposition}\label{2.18}
Let $\mathbf F$ be defined as in (\ref{phyF}).  Then, equality (\ref{Zweak}) holds for any $\varphi \in C_c^1((0, \tau ) \times \Omega ;\mathbb{R}^3)$.  Moreover, 
 we have that $\mathbf Z(\cdot , \mathbf x) \in W^{1, \infty }([0, \tau ))$ for $\sigma _0-a.e.$ $\mathbf x \in \Omega $, and (\ref{2.43}) holds.
\end{proposition}

\begin{proof}
From the definition of the Lagrangian flow $\Phi $ in Lemma \ref{2.8}, we have that
\[ \Phi (t, \mathbf y) = \mathbf y + \int_0^t \! \wt (s, \Phi _{(s)}(\mathbf y)) \, ds\]
for $\nu _0-a.e.$ $\mathbf y \in \Lambda $ and every $t \in [0, \tau )$.  Thus, this equality holds for all $\mathbf y \in \Lambda \setminus N$ where $\nu _0 [N] = 0$.  Since $\mathbf T_0 \textrm{\#}\sigma _0 = \nu _0$ 
 it follows that
\[ \sigma _0 \left [ \Omega \cap \mathbf T_0^{-1}(N) \right ] = \int_{\mathbf T_0^{-1}(N)} \! \sigma _0(\mathbf x) \, d\mathbf x = \int_N \! \nu _0(\mathbf y) \, d\mathbf y = 0.\]
Thus, for $\sigma _0-a.e.$ $\mathbf x \in \Omega $ and every $t \in [0, \tau )$, we have
\begin{equation}\label{2.78}
\Phi (t, \mathbf T_0(\mathbf x)) = \mathbf T_0(\mathbf x) + \int_0^t \! \mathbf w (s, \Phi _{(s)}(\mathbf T_0(\mathbf x))) \, ds,
\end{equation}
where we have replaced $\wt (s, \Phi _{(s)}(\mathbf T_0(\mathbf x)))$ by $\mathbf w (s, \Phi _{(s)}(\mathbf T_0(\mathbf x)))$ based on (\ref{phiinlambda}), (\ref{phiw}).  Multiplying (\ref{2.78}) by $\sigma _0(\mathbf x)$ and by $\partial _t \varphi (t, \mathbf x)$, where $\varphi \in C_c^1([0, \tau ) \times \mathbb{R}^3)$, and then integrating we obtain
\begin{eqnarray*}
 \int_{[0, \tau ) \times \Omega } \! \partial _t \varphi (t, \mathbf x) \Phi (t, \mathbf T_0(\mathbf x)) \sigma _0(\mathbf x) \, dtd\mathbf x&& =  \int_{[0, \tau ) \times \Omega } \! \partial _t \varphi (t, \mathbf x) \mathbf T_0(\mathbf x) \sigma _0(\mathbf x) \, dtd\mathbf x\\
 && +  \int_{[0, \tau ) \times \Omega } \! \sigma _0(\mathbf x)\partial _t \varphi (t, \mathbf x)\int_0^t \! \mathbf w (s, \Phi _{(s)}(\mathbf T_0(\mathbf x))) ds \, dtd\mathbf x.
 \end{eqnarray*}
Now, in the right-hand side, we perform the integration with respect to $t$ in the first integral and integrate by parts with respect to $t$ in the second integral to obtain
 
\begin{eqnarray*}
 \int_{[0, \tau ) \times \Omega } \! \partial _t \varphi (t, \mathbf x) \Phi (t, \mathbf T_0(\mathbf x)) \sigma _0(\mathbf x) \, dtd\mathbf x&& =  -\int_{\Omega } \! \varphi (0, \mathbf x) \mathbf T_0(\mathbf x) \sigma _0(\mathbf x) \, dtd\mathbf x\\
 && -  \int_{[0, \tau ) \times \Omega } \! \varphi (t, \mathbf x)\mathbf w (t, \Phi _{(t)}(\mathbf T_0(\mathbf x))) \sigma _0(\mathbf x) \, dtd\mathbf x,
 \end{eqnarray*}
\begin{equation}\label{2.79}
\end{equation}
where we have used that $\varphi (\tau , \mathbf x) \equiv 0$ due to its compact support.
Note that, by (\ref{comdual2}), (\ref{phyF}) and Lemma \ref{2.16}, we have
\begin{equation}\label{2.80}
\mathbf w(t, \Phi _{(t)}(\mathbf T_0(\mathbf x))) = \mathbf e_3 \times \left [ \Phi _{(t)}(\mathbf T_0(\mathbf x)) - \mathbf S_{(t)}(\Phi _{(t)}(\mathbf T_0(\mathbf x)))\right ] = \mathbf e_3 \times \left [ \mathbf Z(t, \mathbf x) - \mathbf F(t, \mathbf x)\right ]
\end{equation}
for $\sigma _0-a.e.$ $\mathbf x \in \Omega $ and every $t \in [0, \tau )$.  Substituting (\ref{2.80}) into the right-hand side of (\ref{2.79}) and using Lemma \ref{2.16} to replace $\Phi (t, \mathbf T_0(\mathbf x))$ by $\mathbf Z(t, \mathbf x)$ in the left-hand side of (\ref{2.79}), we obtain

\begin{eqnarray*}
\int_{[0, \tau ) \times \Omega } \! \partial _t \varphi (t, \mathbf x) \mathbf Z(t, \mathbf x)\sigma _0(\mathbf x) \, dtd\mathbf x&& =  -\int_{\Omega } \! \varphi (0, \mathbf x) \mathbf T_0(\mathbf x) \sigma _0(\mathbf x) \, dtd\mathbf x\\
 && -  \int_{[0, \tau ) \times \Omega } \! \varphi (t, \mathbf x)\mathbf e_3 \times \left [ \mathbf Z(t, \mathbf x) - \mathbf F(t, \mathbf x)\right ] \sigma _0(\mathbf x) \, dtd\mathbf x,
  \end{eqnarray*}
and rearranging gives (\ref{Zweak}).

Finally, $\mathbf Z(\cdot , \mathbf x) \in W^{1, \infty }([0, \tau ))$ for $\sigma _0-a.e.$ $\mathbf x \in \Omega $ 
follows from Lemma \ref{2.16} and Lemma \ref{2.8} (\emph{i}).  Then, (\ref{2.78}), (\ref{2.80}) and Lemma \ref{2.16} imply (\ref{2.43}).
\end{proof}

Now the properties of $(\mathbf T, \sigma )$ in Theorem \ref{5.5} and the properties of $\mathbf F$ proved in Propositions \ref{2.13}, \ref{2.14}, \ref{2.15}, \ref{2.17}, \ref{2.18} imply Theorem \ref{2.2}.

\section{Conclusion}
The main result of this paper is the proof of existence of  weak Lagrangian solutions of the fully compressible semi-geostrophic equations with rigid boundary conditions, in the original formulation with variables expressing physically relevant quantities. This result is stated in Theorem \ref{2.2}, and can be considered as the conclusion of the analysis of the problem given in \cite{maroofi}, where an existence result is proved but only in the so-called dual formulation. We have also proved that, if additional regularity of the flow could be assumed, this weak Lagrangian solution would determine a weak (Eulerian) solution of these equations.

In addition to the main result, we have given an alternative proof, based on recent results of Ambrosio and Gangbo on Hamiltonian ODEs in spaces of probability measures, of the previous result on the existence of weak solutions of the dual formulation of the equations.

\section{Acknowledgements}
We wish to thank the anonymous referee for pointing out a serious error in the first version of this paper, and for many useful suggestions. 

DKG gratefully acknowledges the support of an EPSRC-CASE studentship.

\bibliographystyle{acm}
\bibliography{semigeostrophicbib}

\end{document}